\documentclass[reprint, amsmath, amssymb, aps]{revtex4-2}

\usepackage{hyperref}
\hypersetup{colorlinks=true,linkcolor=blue,citecolor = blue, urlcolor=black}

\usepackage{graphicx}
\usepackage{dcolumn}
\usepackage{bm}
\usepackage[utf8]{inputenc}
\usepackage{amssymb}
\usepackage{amsthm}
\usepackage{mathrsfs}
\usepackage{dsfont}
\usepackage{multirow}
\usepackage{ upgreek }
\usepackage{enumitem}
\usepackage{ stmaryrd } 
\usepackage{float}
\usepackage{physics}
\usepackage{cancel}
\usepackage{array}
\usepackage{longtable}

\newcommand{\R}{\mathbb{R}}

\newcommand{\C}{\mathbb{C}}

\newcommand{\1}{\mathds{1}}

\newcommand{\x}{\times}
\renewcommand{\bar}{\overline}

\newcommand{\mycomment}[1]{}

\newtheorem{thm}{Theorem}

\newtheorem{lem}{Lemma}
\newtheorem{conj}{Conjecture}
\newtheorem{prop}{Proposition}

\begin{document}

\preprint{APS/123-QED}

\title{On The Stabilizer Formalism And Its Generalization}

\author{\'Eloi DESCAMPS$^{1,2}$, Borivoje DAKI\'C$^{3,4}$}
\affiliation{$^1$Département de Physique de l’\'Ecole Normale Supérieure - PSL, 45 rue d’Ulm, 75230 Paris, France}
\affiliation{$^2$Universitée Paris Cité, Laboratoire Matériaux et Phénomènes Quantiques, 75013 Paris, France}
\affiliation{$^3$University of Vienna, Faculty of Physics, Vienna Center for Quantum
Science and Technology, Boltzmanngasse 5, 1090 Vienna, Austria}
\affiliation{$^4$ Institute for Quantum Optics and Quantum Information (IQOQI) Vienna,
Austrian Academy of Sciences, Boltzmanngasse 3, A-1090 Vienna, Austria}

\date{\today}

\begin{abstract}
The standard stabilizer formalism provides a setting to show that quantum computation
restricted to operations within the Clifford group are classically efficiently simulable: this is the
content of the well-known Gottesman-Knill theorem \cite{gottesman_heisenberg_1998}. This work analyzes the mathematical
structure behind this theorem to find possible generalizations and derivation of constraints
 required for constructing a non-trivial generalized Clifford
group. We prove that if the closure of the stabilizing set is dense in the set of $SU(d)$ transformations, then the associated Clifford group is trivial, consisting only of local gates and permutations of subsystems. This result demonstrates the close relationship between the density of the stabilizing set and the simplicity of the corresponding Clifford group. We apply the analysis to investigate stabilization with binary observables for qubits and find that the formalism is equivalent to the standard stabilization for a low number of qubits.  Based on the observed patterns, we conjecture that a large
class of generalized stabilizer states are equivalent to the standard ones. Our results can be used
to construct novel Gottesman-Knill-type results and consequently draw a sharper line between
quantum and classical computation.
\end{abstract}

\maketitle

\section{Introduction}\label{section: intro}
Quantum computation uses quantum systems to perform calculations beyond the capabilities of classical (standard) computers \cite{nielsen_quantum_2010}. Many quantum algorithms solve problems seemingly intractable for classical computers. Prominent examples include Shor's factoring \cite{shor_polynomial-time_1997}, Grover search \cite{grover_fast_1996}, phase estimation \cite{kitaev_quantum_1995}, quantum simulation algorithms \cite{feynman_simulating_1982} etc. A common belief is that only some parts of quantum
computation possess a significant speedup. Thus, it is crucial to identify sets of gates that make quantum computation classically (in)tractable to
understand how this advantage arises. Some models do not bring (exponential) speedups, and this question has been extensively studied in
the literature \cite{anders_fast_2006,vidal_efficient_2003} prominently with the aid of the so-called
\textit{stabilizer formalism}. Based on this formalism, the Gottesman-Knill theorem
identifies a“non-trivial” portion of the quantum-mechanical
computation that can be efficiently simulated classically.
The theorem has a pedagogical value in showing
substantial differences and similarities between classical
and quantum computation. This theorem also saw few
generalizations obtained by varying different
parameters of the setting \cite{clark_generalised_2007,bermejo-vega_normalizer_2015} leading to a
better insights into the boundary between classical and
quantum models.

In this work, we continue this search for
possible generalizations of the Gottesman-Knill theorem
and the underlying stabilizer formalism. Our main
results are: a) we prove the constraint theorem for the generalized Clifford group,  b) we provide an exhaustive analysis of binary stabilization for two and three qubits, and based on observed patterns, c) we conjecture that binary stabilization for qubits is equivalent to the standard stabilizer formalism. The paper is organized as follows. In Section~\ref{section: stab_formalism} we discuss in more details the stabilizer formalism and the Gottesman-Knill theorem and then proceed by defining the generalized stabilizer states, stabilizer sets and Clifford groups. In Section~\ref{section: cliff_group}, we focus on infinite stabilization sets, specifically those that generate a dense set within the group of $SU(d)$ transformations. In such case, we derive a theorem that demonstrates that the corresponding Clifford group is trivial, composed solely of local gates and permutations of subsystems. This finding highlights the direct relationship between the density of the stabilization set and the triviality of the associated Clifford group. Lastly, in Section~\ref{section: gen_stab} we discuss the stabilization for qubits and binary operators. We provide a complete analysis for the two and three qubits cases and we show that in these cases the stabilizer formalism is equivalent to the standard one. Based on this, we conjecture that this equivalence extends to an arbitrary number of qubits.

\section{Towards  a generalization of the stabilizer formalism}\label{section: stab_formalism}
We start our analysis with the set of $N$ qudits residing in the Hilbert space $(\C^d)^{\otimes N}$. The basic idea of the stabilizer formalism is to fix a state $\ket{\psi}$ with the set of operators $O_1, \cdots,O_k$ such that:
\begin{itemize}
	\item $\ket{\psi}$ is a $+1$-eigenvector (with eigenvalue $+1$) for all $O_i$, and
	\item $\ket{\psi}$ is unique (up to a complex factor).	
\end{itemize}
These hypotheses allow for a full characterization of the state $\ket{\psi}$ by only specifying this list $O_1,\cdots,O_k$ of operators. We will consider (local) 
tensor operators $O_i\in \mathcal{A}^{\otimes N}$ with $\mathcal{A}\subset{U(d)}$, where $U(d)$ is the group of $d\times d$ unitary matrices. The set $\mathcal{A}$ is called the \textit{stabilizing set}, and any state $\ket{\psi}$ uniquely stabilized (\textit{i.e.}, satisfying the two conditions above) is a \textit{stabilizer state}.\\

We now introduce the \textit{generalized Clifford group} $\mathcal{C}_N(\mathcal{A})$ as follows
\begin{equation}\label{eq: def clif group}
	\mathcal{C}_N(\mathcal{A})=\{U\in U(d^N)|\forall O\in \mathcal{A}^{\otimes N},~UOU^{\dagger}\in\mathcal{A}^{\otimes N}\}.
\end{equation}
The definition ensures that $O\ket{\psi}$ is also a stabilizer state for any stabilizer states $\ket{\psi}$ and $O\in\mathcal C_N(\mathcal A)$.\\

A classic example of such stabilization setting is given by the set of Pauli matrices and this is the (standard) stabilizer formalism \cite{gottesman_stabilizer_1997}. In this case $\mathcal A= \mathcal{P}$, and

\begin{equation}
    \mathcal{P}=\{\pm 1,\pm i\}\cdot\{\1,\sigma_X,\sigma_Y,\sigma_Z\},
\end{equation}
where $\1$ is the identity and $\sigma_X$, $\sigma_Y$ and $\sigma_Z$ are the usual Pauli matrices. The corresponding Clifford group is:
\begin{equation}
	\mathcal{C}_N=\{U\in U(d^N)|\forall O\in \mathcal{P},~UOU^{\dagger}\in\mathcal{P}\}.
\end{equation}
This Clifford group is generated by two single-qubit gates and one two-qubit gate \cite{selinger_generators_2015}:
\begin{itemize}
	\item Hadamard gate: $H=\frac{1}{\sqrt{2}}\left(\begin{array}{cc} 1 & 1 \\ 1 & -1 \\\end{array}\right)$
	\item $\frac{\pi}{2}$-phase gate: $S= \left(\begin{array}{cc} 1 & 0 \\ 0 & i \\\end{array}\right)$
	\item Controlled phase gate $\Lambda Z =\left(\begin{array}{cccc} 1 & 0 & 0 & 0 \\ 0 & 1 & 0 & 0 \\0 & 0 & 1 & 0\\0 & 0 & 0 & -1\end{array}\right)$ or \item Controlled not gate $\Lambda X =\left(\begin{array}{cccc} 1 & 0 & 0 & 0 \\ 0 & 1 & 0 & 0 \\0 & 0 & 0 & 1\\0 & 0 & 1 & 0\end{array}\right)$
\end{itemize}
Computation involving gates from the Clifford group can be realised as a successive application of these three generators.~With  these definitions we can give a precise statement of the Gottesman-Knill theorem \cite{gottesman_heisenberg_1998}:

\begin{thm}\label{thm: gottesman knill} (Gottesman-Knill)
Computation utilizing only: 
\begin{itemize}
	\item Preparation of qubits in states of the computational basis,
	\item Quantum gates from the Clifford group $\mathcal{C}_N$, and
	\item Measurements in the computational basis,
\end{itemize}
can be simulated efficiently on a classical probabilistic computer.
\end{thm}

The proof  is mainly based on the stabilizer formalism \cite{nielsen_quantum_2010}, suggesting that the general definition we gave for the stabilization may lead to a generalized version of the Gottesman-Knill theorem. However, it may happen that the set $\mathcal C_N(\mathcal A)$ defined in eq. (\ref{eq: def clif group}),
is composed of local gates only, reducing the model to computation by local gates only. In such situation, the corresponding generalized theorem would be trivially true. Thus, two interesting questions arise in this respect:
\begin{itemize}
	\item For which set $\mathcal A$ is the corresponding Clifford group not trivial?
	\item What kind of states can be stabilized by different sets $\mathcal{A}$?
\end{itemize}
What is meant by a trivial Clifford group here is a
group contained in $U(d)^{\otimes N}P_N(d)$, where $P_N(d)$ is the set of permutation gates on  $N$ qudits, with elements $P_\sigma$ acting on the basis states in the following way
\begin{equation}
    P_\sigma\ket{i_1,\cdots,i_N}=\ket{i_{\sigma(1)},\cdots,i_{\sigma(N)}},
\end{equation}
for any permutation $\sigma$ of $N$ elements. In other words, a trivial Clifford group is composed only of local and permutation gates. Such a group does not generate entanglement and thus the computation involving only gates from it is considered trivial.

\section{Generalized Clifford Group}\label{section: cliff_group}
We shall try to find examples of $\mathcal{C}_N(\mathcal A)$ being non-trivial, this means containing entangling gates.  This problem was also studied by \cite{clark_generalised_2007} where the authors studied the Clifford group with $\mathcal{A}$ being a finite group acting irreducibly on $\C^d$. Here we consider a more general case of infinite groups being involved. We cover a first step in this direction by considering the stabilizing set $\mathcal A$ to generate a group $\langle \mathcal{A}\rangle$ dense inside $SU(d)$, \textit{i.e.}, $SU(d)\subset \overline{\langle \mathcal A\rangle}$, where $\overline O$ denote the topological closure of the set $O$.\\
Under these assumptions we show a series of results leading to a new necessary condition for a generalized Clifford group to be non-trivial. We show the following result.
\begin{prop}\label{pro: adherence clifford}
    If the set $\langle\mathcal A\rangle$ is dense inside $SU(d)$, then for all integers $N$, the Clifford group $\mathcal{C}_N(\mathcal{A})$ verifies $\mathcal{C}_N(\mathcal{A})\subset \mathcal{C}_N(U(d))$.
 \end{prop}
\begin{proof}[\underline{Proof}]
See Appendix (\ref{annex: proof pro 1})
\end{proof}

Next, we show that a gate from $\mathcal{C}_N(U(d))$ is non-entangling, which means that it maps product-states to product-states.

\begin{prop}\label{pro: non entangling}
 If $U\in\mathcal{C}_N(U(d))$, then for all product-state $\ket{\psi}=\ket{\varphi_1}\cdots\ket{\varphi_N}$, we have $U\ket{\psi}=\ket{\tilde{\varphi_1}}\cdots\ket{\tilde{\varphi_N}}$ for some states $\ket{\tilde{\varphi_1}},\cdots,\ket{\tilde{\varphi_N}}$.
\end{prop}

\begin{proof}[\underline{Proof}]
See Appendix (\ref{annex: proof pro 2})
\end{proof}
 An equivalent result can be obtained for product ``\textit{bra}" (dual) states. This is simply obtained by varying the proof and looking at the left eigenvectors. We thus have that all $U\in \mathcal C_N(\mathcal A)$ stabilize the set of product linear forms, \textit{i.e.}, $\bra{\psi}=\bra{\varphi_1}\cdots\bra{\varphi_N}$, and $\bra{\psi}U=\bra{\tilde{\varphi_1}}\cdots\bra{\tilde{\varphi_N}}$.\\
With this, we can show an already known result \cite{brylinski_universal_2001}, of factorizing non-entangling gates into a tensor product of local gates up to a permutation gate. However, in comparison to the proof provided in \cite{brylinski_universal_2001} our proof employs elementary mathematics and provides the result in full generality (for arbitrary $d$ and $N$).

\begin{prop}\label{pro: non entangling are trivial}
The non-entangling set $\mathcal{C}_N(U(d))$ acting on N qudits is composed of trivial gates , \textit{i.e.},
\begin{equation}
    \mathcal C_N(U(d))=U(d)^{\otimes N}P_N(d).
\end{equation}
\end{prop}

\begin{proof}[\underline{Proof}]
See Appendix (\ref{annex: proof prop 3}).
\end{proof}

By combining these intermediary results, we arrive to the following theorem:
\begin{thm}\label{thm: trivial clifford group}
    Let $\mathcal{A} \subset U(d)$ such that, the group $\langle \mathcal{A}\rangle$ generated by $\mathcal{A}$ is dense inside $SU(d)$. Then for all integers $N$, the Clifford group of order $N$, $\mathcal{C}_N(\mathcal{A})$ is trivial, \textit{i.e.}, verifies $\mathcal{C}_N(\mathcal{A})\subset U(d)^{\otimes N}P_N(d)$.
\end{thm}

Theorem~\ref{thm: trivial clifford group} simply means that if the set $\mathcal{A}$ generates a group which is too big, then its corresponding generalized Clifford group $\mathcal{C}_N(\mathcal{A})$ is only made of trivial gates. In this situation, a computation starting with a product-state will keep a separable form throughout the application of gates from the Clifford group. Given this, we can keep track of each qudit individually and perform an efficient classical simulation of the computation. This result gives a non-trivial constraint on the stabilizing set $\mathcal{A}$, which is needed to generalize the stabilizer formalism.

\section{Exploring generalized stabilization for qubits}\label{section: gen_stab}
In the previous section we found some constraints on the stabilizing sets $\mathcal A$. However, the Gottesman-Knill theorem does not only rely on the fact that the standard Clifford group $\mathcal C_N$ is not trivial but also that the set of all Pauli-stabilizer states (\textit{i.e.}, state of the standard stabilizer formalism) posses a specific and rich structure. Thus, when analyzing a potential generalized stabilizer setting, an important question to address is determining the specific nature ({\it e.g.}, entanglement structure) of the stabilizer states.
 We focus our attention to the case of qubits ($d = 2$). To arrive at novel structures, it is thus important 
to understand when a set $\mathcal{A}$ stabilizes states
that are not locally equivalent to Pauli-stabilized ones.
 We say that two $N$-qubits states $\ket{\psi}$ and $\ket{\phi}$ are locally equivalent if there exist a local gate $U=U_1\otimes\cdots\otimes U_N$ such that $\ket{\psi}=U\ket{\phi}$ with $U_i\in U(2)$.

\subsection{Binary operator case}\label{section: binary case}
The problem in its full generality ($\mathcal A$ composed of arbitrary unitary matrices) is hard, we will thus consider a simpler case of binary observables. This choice comes naturally, as the Pauli matrices generating the Pauli group (standard stabilizer formalism), are specific instance of binary operators, \textit{i.e.}, operators of the form $\va{\sigma}\cdot \va{n}$ where $\va{\sigma}=(\sigma_X,\sigma_Y,\sigma_Z)$ and $\va{n}$ is a unit-norm vector. We first take a look at the stabilization problem when the elements of $\mathcal A$ are general binary operators.\\

Recall that a necessary condition to have stabilization by a list of operators, is to have for each pair of stabilizers common $+1$-eigenvectors. We thus investigate this pairwise condition in the case of binary observables $\va{\sigma}\cdot\va{n}$. Since stabilized sets are equivalent up to local unitary, without loss of generality, we set the first stabilizing operator to be $\sigma_Z\otimes\cdots\otimes\sigma_Z$ and the second one can to be composed of operators in the $(XZ)$-plane (\textit{i.e.}, operators of the form $\cos(\theta)\sigma_Z+\sin(\theta)\sigma_X$). To begin with, we introduce some useful notation.\\

\underline{Notation:}
\begin{itemize}
	\item For real angles $\theta$ and $\phi$ a general binary operator on the Bloch sphere is defined as
 \begin{equation}\label{eq: def opera A theta phi}
     A_{\theta,\phi}=\cos(\theta)\sigma_Z+\sin(\theta)\Big(\cos(\phi)\sigma_X+\sin(\phi)\sigma_Y\Big)
 \end{equation}
	\item For operators in the $(XZ)$-plane we set $A_\theta=A_{\theta,0}$.
	\item For $n$ binary operators $A_1$,...,$A_n$ the projector on the $+1$-eigenspace is denoted as: $P_{A_1\otimes\cdots\otimes A_n}=\frac{1}{2}(\1+A_1\otimes\cdots\otimes A_n)$
	\item For operators in the $(XZ)$-plane, $A_{\theta_1}$,...,$A_{\theta_n}$ we set  notation $P_{A_{\theta_1}\otimes\cdots\otimes A_{\theta_n}}=P_{\theta_1,\dots,\theta_n}$
	\item For indices $j_1,\dots,j_n\in\{-1,+1\}$, we define the unormalized states:
    \begin{equation}
        \label{equstate}
        \ket{\psi_{j_1,\dots,j_n}}=\operatorname{Re}\left[\sum\limits_{k_1,\dots,k_n=0,1} (i j_1)^{k_1}\cdots(i j_n)^{k_n}\ket{k_1 \cdots k_n}\right].
	\end{equation}
	
\end{itemize}

From a technical standpoint, the last definition proves to be highly useful as in our calculations, it frequently becomes necessary to retain terms within sums that involve an even number of indices $k_i$, all of which are set to the value $1$. We also see that $\ket{\psi_{j_1,\cdots,j_n}}=\ket{\psi_{-j_1,\cdots,-j_n}}$ and this is why we fix the first index to $+1$ to avoid counting twice the same state.
With these definitions in place, we can state a technical theorem which has profound consequences on the stabilization with binary operators in the $(XZ)$-plane.

\begin{thm}\label{thm: eigenvalue of projectors}
 For $n$ angles $\theta_1,\dots,\theta_n$, the eigenvalues (with multiplicity) of $P_{0,\cdots,0} P_{\theta_1,\cdots,\theta_n}$ are the following: 
\begin{itemize}
 	\item  $2^{n-1}$ zeros,
 	\item $\frac{1}{2}(1+\cos(j_1\theta_1+\cdots+j_n\theta_n))$ for $j_1,...,j_n\in\{-1,+1\}$ and $j_1=1$, with the corresponding eigenvector $\ket{\psi_{j_1,\cdots,j_n}}$.
 \end{itemize}
\end{thm}

\begin{proof}[\underline{Proof}]
    See Appendix (\ref{annex: proof thm 6})
\end{proof}

Recall that  we are interested in the stabilization of a state by a family of tensor product of binary operators. Two of them are chosen as a reference, and, as already said, due to local equivalence, we can always set them to $\sigma_Z\otimes\cdots\otimes\sigma_Z$ and $A_{\theta_1}\otimes\cdots\otimes A_{\theta_N}$. Their common $+1$-eigenvector is simply given by a state stabilized by the product of the projectors $P_{0,\cdots,0}P_{\theta_1,\cdots,\theta_N}$. Given Theorem~\ref{thm: eigenvalue of projectors}, we can conclude the following. As the eigenvalues are $\frac{1}{2}(1+\cos(j_1\theta_1+\cdots+j_n\theta_n))$, we can vary parameters $\theta_1,\dots,\theta_n$ to control the number of stabilized states. Note that these eigenstates are independent of the value of $\theta_k$. Thus stabilization with operators in the $(XZ)$-plane, we will only lead to states given by equation (\ref{equstate}).
On the other hand, such states can be stabilized by taking tensor product operator from $\{\sigma_Z,\pm\sigma_X\}$ only. Indeed, $\ket{\psi_{j_1,\cdots,j_N}}$ is uniquely stabilized by the operators
\begin{equation}
    O_k=\sigma_X\otimes \sigma_Z\otimes\cdots\otimes\sigma_Z\otimes\underbrace{(-j_{k}\sigma_X)}_{\text{at }k}\otimes\sigma_Z\otimes\cdots\otimes\sigma_Z,
\end{equation}
for $k=2,\dots, n$. To see this, we can compute the eigenvalue of $O_k$ for any state $\ket{\psi_{i_1,\cdots,i_N}}$ and verify that it is $1$ if and only if $i_k=j_k$.

This means all states stabilized by operators in the $(XZ)$-plane are locally equivalent to Pauli stabilizer states. However, note that not all of Pauli stabilizer states can be obtained in this way.  For example the state $\ket{00}$ cannot be stabilized by only binary operators in the $(XZ)$-plane, because states in eq. (\ref{equstate}) for two qubits are maximally entangled, \textit{i.e}, $\ket{\psi_{j_1,j_2}}=\ket{00}\pm\ket{11}$. Since the identity operator allows for the stabilization of product-states, we shall add the identity as an element of our stabilizing set $\mathcal{A}$.

\subsection{Adding the identity}
We have seen in the previous section that stabilization by operators constructed from the $(XZ)$-plane only yields standard stabilizers states, but states in eq.~(\ref{equstate}) do not span all possibilities, thus adding identity to the set $\mathcal A$ is necessary. Therefore, we will explore stabilizing sets of the form $\mathcal{A}=\{A_{\theta,\phi},\1_2\}$. For a small number of qubits, we can exhaustively search all possibilities, and from there, we can get better insight into general constraints on the stabilizing set. We can also compute the associated stabilizer states up to local equivalence.

\subsubsection{Methods}\label{section: methods}
To do an exhaustive search, firstly, we list all non-equivalent stabilization patterns. A stabilization pattern simply corresponds to the list of stabilizers:
\begin{align}\label{eq: def list stab}
	O_1&=A_{1,1}\otimes \cdots \otimes A_{1,n}\notag\\
	&\cdots\\
	O_k&=A_{k,1}\otimes \cdots \otimes A_{k,n}\notag
\end{align}
We are interested in stabilization up to local rotation, thus, two stabilization patterns are equivalent if they can be transformed to each other via a) permutation of the qubits, b) permutation of the operators, and c) local rotations.
With this, we seek for
\begin{itemize}
 	\item {\bf Unique stabilization.} There is only one common $+1$-
eigenstate of all the operators $O_k$.
 	\item {\bf Minimal stabilization.} We cannot achieve a unique stabilization with a proper subset of the operators $\{O_k\}$.
\end{itemize}

After listing all inequivalent stabilization patterns, we
shall analyze each one separately to identify parameters $\theta_k$ for which we get unique stabilization. Finally, we shall identify stabilized states. To do so we will use two different methods.\\ 

With the first method, which we call the determinant method, we consider $\ket{\psi_i}$ ($1\leq i\leq m$) to be the eigenbasis basis of $O_1$. A common eigenstate $\ket{\psi}$ is then expanded over this basis:
\begin{equation}
    \ket{\psi}=\sum_{i=1}^m \alpha_i\ket{\psi_i}.
\end{equation}
We write $O_j\ket{\psi_i}=\sum_{l=1}^m \beta_{ijl}\ket{\psi_l}$, and thus the condition $O_j\ket{\psi}=\ket{\psi}$ implies for all $1\leq j\leq k$ and $1\leq l\leq m$,
\begin{equation}
    \sum_{i=1}^m \alpha_i \beta_{ijl}=\alpha_l.
\end{equation}
This is a linear system in $m$ unknown $\alpha_i$'s with $km$ equations, which can be summarized in matrix form as:
\begin{equation}
    M\cdot \begin{pmatrix}
        \alpha_1\\\vdots\\\alpha_m
    \end{pmatrix}=0,
\end{equation}
with $M$ being a $km\times m$ matrix. The solution of the system is then given by the kernel $M$. To find the kernel, we impose the constraint $\det(M^\dagger M)=0$ which gives us a necessary condition for the set of coefficients $\beta_{ijl}$ and consequently on the set of angles $\theta_1,\dots,\theta_N$ that parameterize our stabilizers.\\
  
The second approach, called the projector method, involves examining the projector $P_i=(\1+O_i)/2$ on the $+1$-eigenspace for each stabilizing operator $O_i$. We consider their product $M=P_1P_2\dots P_k$. A state $\ket{\psi}$ is stabilized by all $O_i$ if and only if $M\ket{\psi}=\ket{\psi}$, therefore, we search for the eigenvalues and eigenvectors of $M$. The expressions for the eigenvalues of $M$, combined with the requirement that one of them equals $1$, impose constraints on the angles $\theta_1,\dots,\theta_N$. Furthermore computing the eigenbasis of $M$ directly provides the expression of the stabilizer states.~This method requires more computational resources, making it suitable only for simpler scenarios characterized by a limited number of free parameters.\\

These two methods are applied in the analysis for two and three qubits.

\subsubsection{Two-qubits case}
We have identified four inequivalent stabilization patterns with two operators. 
Table $\ref{tab: stab 2 qubit}$ outlines the stabilization patterns along with the necessary conditions on the free parameters to achieve unique stabilization, as well as the corresponding stabilizer states.

\begin{table}[H]
    \centering
    \begingroup
        \setlength{\tabcolsep}{6pt} 
        \renewcommand{\arraystretch}{1.5} 
        \begin{tabular}{|c|c|c|}
        \hline
            Operators & Conditions & Stabilized states\\
            \hline
            $\begin{matrix}
                \sigma_Z\otimes \sigma_Z \\ 
                A_\theta\otimes A_\phi
            \end{matrix} $ & $\theta \pm \phi =0~[2\pi]$ & $\ket{\psi}=\ket{00}\pm\ket{11}$\\
            \hline         $\begin{matrix}
                \sigma_Z\otimes \sigma_Z \\ 
                A_\theta\otimes \1
            \end{matrix}$ & $\theta=0~[\pi]$ & $\ket{\psi}=\ket{00}$ or $\ket{11}$\\
            \hline $\begin{matrix}
                \sigma_Z\otimes \1 \\ 
                A_\theta\otimes \1
    	\end{matrix} 
    	$ & $\theta=0~[2\pi]$ & Non unique \\
            \hline $\begin{matrix}
                \sigma_Z\otimes \1 \\ 
                \1\otimes \sigma_Z
    	\end{matrix}$ & None & $\ket{\psi}=\ket{00}$\\
            \hline
        \end{tabular}
    \endgroup
    \caption{List of non-equivalent stabilization pattern for two qubits together with the associated stabilizer states.}
    \label{tab: stab 2 qubit}
\end{table}

We see that the stabilizer states are either product-states or maximally entangled states. These are all locally equivalent to Pauli stabilizer states.\\

\subsubsection{Three qubits-case}
The three qubit stabilization analysis is more involved, since there are dozen of inequivalent stabilization patterns. We start the analysis of stabilizations with only two operators. This gives us some minimal unique stabilization, and also a two-operators compatibility rule. We enumerate $9$ distinct stabilization patterns, but only one of them yields a unique stabilization, where the stabilizer states correspond to those described in eq. (\ref{equstate}). Additionally, we identify $8$ patterns that stabilize a two-dimensional space, which establishes a compatibility condition used in the subsequent analysis. From there, we also provide an exhaustive search for three operators. Details are provided in Appendix~\ref{annex: 3 qubit analysis}.\\
We found that stabilization is locally equivalent to Pauli stabilizer states. Table~\ref{tab: two op} and~\ref{tab} shown in appendix~\ref{annex: 3 qubit analysis} summarizes the analysis and findings.

\subsubsection{Conjecture}
Our analysis for a low number of qubits shows that all stabilization with binary operators (including identity operator) is locally equivalent to the standard Pauli stabilization. 
 This motivates us to put forward the following conjecture:
\begin{conj}
    All states stabilized by the set $\mathcal A$ composed of binary operators and the identity {\it i.e.}, $\mathcal{A}=\{A_{\theta,\phi},\1_2\}$, are locally equivalent to standard stabilizer states (\textit{i.e.}, stabilized by the Pauli group $\mathcal{P}=\{\pm 1,\pm i\}\cdot\{\1,\sigma_X,\sigma_Y,\sigma_Z\}$).
\end{conj}

One way to tackle this conjecture is to go back to stabilization
on the $(XZ)$-plane, where a weaker conjecture with $\mathcal{A}=\{A_\theta,\1\}$
is also probably true. Furthermore, on a small number
of qubits, we observe that all Pauli stabilizations also seem possible (up to local equivalence) without the
operator $\sigma_Y$. So, potentially, it is possible that by including the identity matrix as a stabilizer, a way can be found to eliminate the $Y$ components.\\

It is interesting to point out that this conjecture is false for the case of stabilization with general $2\times 2$ unitary matrices. Indeed, as shown in \cite{ni_non-commuting_2015}, there is a setting - the so called \textit{XS-stabilizer formalism} - where one can explicitly construct a stabilizer state on $6$ qubits which is not locally equivalent to a Pauli stabilizer state. The set $\mathcal{A}$ of this setting is the group generated by three operators $\{\sigma_X, \operatorname{diag}(1,i), e^{i\pi/4}\1\}$. The stabilizer state is given by
\begin{equation}
	\ket{\psi}=\sum\limits_{x_i=0,1} (-1)^{x_1x_2x_3}\ket{x_1,x_2,x_3,x_1\oplus x_2,x_2\oplus x_3,x_3\oplus x_1},
\end{equation}
(where $\oplus$ correspond to the addition modulo $2$) and it is stabilized by the following set

\begin{align}
\begin{split}
		O_1&=\sigma_X\otimes S^3\otimes S^3\otimes S \otimes \sigma_X\otimes \sigma_X\\
		O_2&=S^3\otimes \sigma_X\otimes  S^3 \otimes \sigma_X\otimes S\otimes \sigma_X\\
		O_3&=S^3\otimes  S^3 \otimes\sigma_X\otimes \sigma_X\otimes \sigma_X\otimes S.
\end{split}
\end{align}
One can see that this state is not a Pauli stabilizer state because the exponents of the coefficients $(-1)^{x_1x_2x_3}$ are a cubic polynomial in variables $x_1,x_2,x_3$. However, showing that this state is not locally equivalent to a Pauli stabilizer state is more involved and requires a full classification of standard stabilizer formalism states \cite{ni_non-commuting_2015}.

\section{Conclusion}\label{section: conclusion}
We have investigated two important
questions about stabilizer formalism in this work. We 
have found a set of new constraints on the set of
stabilization operators used to
generate non-trivial generalized Clifford groups. We
have also shown that such a set cannot generate a group
dense inside $SU(d)$. Finally, we explored the case of
stabilization with binary operators, and we observed, for a low number of qubits possibility of stabilization only by the standard Pauli group (up to local equivalence). This led us to conjecture that the same holds for any number of qubits.
It is worth pointing out that there exist many generalizations for the case of
qudits ($d > 2$) \cite{gheorghiu_standard_2014,grassl_efficient_2003}. It is thus interesting to investigate equivalent versions of this conjecture in such settings. Our findings offer promising perspectives regarding new Gottesman-Knill-like theorems and potential applications in quantum error correction \cite{gottesman_heisenberg_1998}.

\section{Acknowledgments}\label{section: acknoledgments}
Authors would like to  thank Sebastian Horvat for his valuable feedback. E.D. extends thanks to the University of Vienna and the Erasmus program for enabling his fellowship and contributing to his stay in Vienna. This research was funded in whole, or in part, by the Austrian Science Fund (FWF) [F7115] (BeyondC). For the purpose of open access, the author(s) has applied a CC BY public copyright licence to any Author Accepted Manuscript version arising from this submission.

\bibliographystyle{unsrturl} 
\bibliography{refs}

\begin{thebibliography}{10}

\bibitem{gottesman_heisenberg_1998}
Daniel Gottesman.
\newblock The {Heisenberg} {Representation} of {Quantum} {Computers}, July
  1998.
\newblock Number: arXiv:quant-ph/9807006 arXiv:quant-ph/9807006.
\newblock URL: \url{http://arxiv.org/abs/quant-ph/9807006}.

\bibitem{nielsen_quantum_2010}
Michael~A. Nielsen and Isaac~L. Chuang.
\newblock {\em Quantum computation and quantum information}.
\newblock Cambridge University Press, Cambridge ; New York, 10th anniversary ed
  edition, 2010.

\bibitem{shor_polynomial-time_1997}
Peter~W. Shor.
\newblock Polynomial-{Time} {Algorithms} for {Prime} {Factorization} and
  {Discrete} {Logarithms} on a {Quantum} {Computer}.
\newblock {\em SIAM Journal on Computing}, 26(5):1484--1509, October 1997.
\newblock arXiv:quant-ph/9508027.
\newblock URL: \url{http://arxiv.org/abs/quant-ph/9508027}, \href
  {https://doi.org/10.1137/S0097539795293172}
  {\path{doi:10.1137/S0097539795293172}}.

\bibitem{grover_fast_1996}
Lov~K. Grover.
\newblock A fast quantum mechanical algorithm for database search, November
  1996.
\newblock Number: arXiv:quant-ph/9605043 arXiv:quant-ph/9605043.
\newblock URL: \url{http://arxiv.org/abs/quant-ph/9605043}.

\bibitem{kitaev_quantum_1995}
A.~Yu Kitaev.
\newblock Quantum measurements and the {Abelian} {Stabilizer} {Problem},
  November 1995.
\newblock Number: arXiv:quant-ph/9511026 arXiv:quant-ph/9511026.
\newblock URL: \url{http://arxiv.org/abs/quant-ph/9511026}.

\bibitem{feynman_simulating_1982}
Richard~P. Feynman.
\newblock Simulating physics with computers.
\newblock {\em International Journal of Theoretical Physics}, 21(6):467--488,
  June 1982.
\newblock \href {https://doi.org/10.1007/BF02650179}
  {\path{doi:10.1007/BF02650179}}.

\bibitem{anders_fast_2006}
Simon Anders and Hans~J. Briegel.
\newblock Fast simulation of stabilizer circuits using a graph state
  representation.
\newblock {\em Physical Review A}, 73(2):022334, February 2006.
\newblock arXiv:quant-ph/0504117.
\newblock URL: \url{http://arxiv.org/abs/quant-ph/0504117}, \href
  {https://doi.org/10.1103/PhysRevA.73.022334}
  {\path{doi:10.1103/PhysRevA.73.022334}}.

\bibitem{vidal_efficient_2003}
Guifre Vidal.
\newblock Efficient classical simulation of slightly entangled quantum
  computations.
\newblock {\em Physical Review Letters}, 91(14):147902, October 2003.
\newblock arXiv:quant-ph/0301063.
\newblock URL: \url{http://arxiv.org/abs/quant-ph/0301063}, \href
  {https://doi.org/10.1103/PhysRevLett.91.147902}
  {\path{doi:10.1103/PhysRevLett.91.147902}}.

\bibitem{clark_generalised_2007}
Sean Clark, Richard Jozsa, and Noah Linden.
\newblock Generalised {Clifford} groups and simulation of associated quantum
  circuits, January 2007.
\newblock Number: arXiv:quant-ph/0701103 arXiv:quant-ph/0701103.
\newblock URL: \url{http://arxiv.org/abs/quant-ph/0701103}.

\bibitem{bermejo-vega_normalizer_2015}
Juan Bermejo-Vega, Cedric Yen-Yu Lin, and Maarten Van~den Nest.
\newblock Normalizer circuits and a {Gottesman}-{Knill} theorem for
  infinite-dimensional systems, January 2015.
\newblock arXiv:1409.3208 [math-ph, physics:quant-ph].
\newblock URL: \url{http://arxiv.org/abs/1409.3208}.

\bibitem{gottesman_stabilizer_1997}
Daniel Gottesman.
\newblock Stabilizer {Codes} and {Quantum} {Error} {Correction}, May 1997.
\newblock arXiv:quant-ph/9705052.
\newblock URL: \url{http://arxiv.org/abs/quant-ph/9705052}.

\bibitem{selinger_generators_2015}
Peter Selinger.
\newblock Generators and relations for n-qubit {Clifford} operators.
\newblock {\em Logical Methods in Computer Science}, 11(2):10, June 2015.
\newblock arXiv:1310.6813 [quant-ph].
\newblock URL: \url{http://arxiv.org/abs/1310.6813}, \href
  {https://doi.org/10.2168/LMCS-11(2:10)2015}
  {\path{doi:10.2168/LMCS-11(2:10)2015}}.

\bibitem{brylinski_universal_2001}
Jean-Luc Brylinski and Ranee Brylinski.
\newblock Universal quantum gates, August 2001.
\newblock Number: arXiv:quant-ph/0108062 arXiv:quant-ph/0108062.
\newblock URL: \url{http://arxiv.org/abs/quant-ph/0108062}.

\bibitem{ni_non-commuting_2015}
Xiaotong Ni, Oliver Buerschaper, and Maarten Van~den Nest.
\newblock A {Non}-{Commuting} {Stabilizer} {Formalism}.
\newblock {\em Journal of Mathematical Physics}, 56(5):052201, May 2015.
\newblock arXiv:1404.5327 [cond-mat, physics:quant-ph].
\newblock URL: \url{http://arxiv.org/abs/1404.5327}, \href
  {https://doi.org/10.1063/1.4920923} {\path{doi:10.1063/1.4920923}}.

\bibitem{gheorghiu_standard_2014}
Vlad Gheorghiu.
\newblock Standard {Form} of {Qudit} {Stabilizer} {Groups}.
\newblock {\em Physics Letters A}, 378(5-6):505--509, January 2014.
\newblock arXiv:1101.1519 [quant-ph].
\newblock URL: \url{http://arxiv.org/abs/1101.1519}, \href
  {https://doi.org/10.1016/j.physleta.2013.12.009}
  {\path{doi:10.1016/j.physleta.2013.12.009}}.

\bibitem{grassl_efficient_2003}
Markus Grassl, Martin Roetteler, and Thomas Beth.
\newblock Efficient {Quantum} {Circuits} for {Non}-{Qubit} {Quantum}
  {Error}-{Correcting} {Codes}.
\newblock {\em International Journal of Foundations of Computer Science},
  14(05):757--775, October 2003.
\newblock arXiv:quant-ph/0211014.
\newblock URL: \url{http://arxiv.org/abs/quant-ph/0211014}, \href
  {https://doi.org/10.1142/S0129054103002011}
  {\path{doi:10.1142/S0129054103002011}}.

\end{thebibliography}

\appendix
\onecolumngrid
\section{Proof of proposition~\ref{pro: adherence clifford}}
\label{annex: proof pro 1}
\begin{proof}[\underline{Proof}] 	
 	$\blacktriangleright$ We start with $N=2$ qudits and a gate $U\in \mathcal{C}_2(\mathcal{A})$.  We define the set
 	\begin{equation}
 		G=\{u\otimes v\in U(d)\otimes U(d)|U(u\otimes v)U^\dagger \in U(d)\otimes U(d)\}.
 	\end{equation}
We show that $G=U(d)\otimes U(d)$. This is because of the following:
\begin{itemize}
    \setlength\itemsep{0em}
	\item G is a subgroup of $U(d)\otimes U(d)$,
	\item By definition of $U$, tensor product of elements of $\mathcal{A}$ are in $G$ , {\it i.e.},  $\mathcal{A} \otimes \mathcal{A}\subset G$,
	\item $G$ is closed,
	\item $G$ remains stable when multiplied by a complex phase, {\it i.e.}, $\forall \varphi,~e^{i\varphi}G\subset G$,
	\item Since $SU(d)\subset \bar{\langle \mathcal{A}\rangle}$, for all $u\in U(d)$, we have $u\otimes u\in G$.
\end{itemize}
Thus, for $A_1,A_2\in\mathcal{A}$, $u,w\in U(d)$ and $e^{i\varphi}\in U(1)$, we have:
\begin{align}
\begin{split}
	&e^{i\varphi}(uA_1w)\otimes(uA_2w)\in G\\
	&\Rightarrow e^{i\varphi}(uA_1A_2^\dagger u^\dagger)\otimes \1\in G ~~(\text{ for }w=(uA_2)^\dagger ).
\end{split}
\end{align} 
	
 Now we introduce
\begin{equation}
 H=\bar{\langle \{e^{i\varphi}uA_1A_2^\dagger u^\dagger|\varphi\in\R,~u\in U(d),~A_1,A_2\in \mathcal{A}\}\rangle}.
 \end{equation} 
 $H$ is a closed normal subgroup of $U(d)$ containing $U(1)$, by definition. For the subgroup $\tilde{H}$ of $H$ of determinant-one matrices, we see that $\tilde H\triangleleft SU(d)$. Because $SU(d)$ is a simple Lie group, we only have two possibilities (since $\tilde{H}$ is also closed):
 \begin{itemize}
    \setlength\itemsep{0em}
 	\item $\tilde{H}\subset \operatorname{Z}(SU(d))=\{e^{2i\pi k/d}\1|k=0,\cdots,d-1\}$ (the center of $SU(d)$). This would imply $H= U(1)$, and thus for all $A_1,A_2\in\mathcal{A}$, $A_1\propto A_2$. This is not possible since $\mathcal{A}$ generate the entire $SU(d)$ group.
 	\item Because of the previous point, we necessarily have $\tilde{H}=SU(d)$ and $H=U(d)$.
 \end{itemize}
 This means that $U(d)\otimes \1\subset G$. Similarly, we show that $\1\otimes U(d)\subset G$. And we finally get  $G=U(d)\otimes U(d)$.  Thus $U\in \mathcal{C}_2(U(d))$ and $\mathcal{C}_2(\mathcal{A})\subset \mathcal{C}_2(U(d))$.\\
 
 $\blacktriangleright$ For the general case of $N$ qudits, we take $U\in \mathcal{C}_N(\mathcal{A})$ and we define in a similar manner the group $G$:
 \begin{equation}
 	G=\{u_1\otimes\cdots\otimes u_N\in U(d)^{\otimes N}|U(u1\otimes\cdots\otimes u_N )U^\dagger \in U(d)^{\otimes ^N}\}.
 \end{equation}
 
 All the properties listed for $G$ (generalized on $N$ qudits) for the two-qudits case are still valid. This means that for all $A_1,A_2\in\mathcal{A}$, $u,w\in U(d)$, $\varphi\in\R$ we have
 \begin{align}
 \begin{split}
 	&(e^{i\varphi}uA_1w)\otimes\cdots\otimes(e^{i\varphi}uA_1w)\otimes(uA_2w)\in G\\
 	\Rightarrow &(e^{i\varphi}uA_1A_2^\dagger u^\dagger)\otimes\cdots\otimes(e^{i\varphi}uA_1A_2^\dagger u^\dagger)\otimes \1\in G.
 	\end{split}
 \end{align}
As before, we introduce the set: 
\begin{equation}
H=\bar{\langle \{e^{i\varphi}uA_1A_2^\dagger u^\dagger|\varphi\in\R,~u\in U(d),~A_1,A_2\in \mathcal{A}\}\rangle}
\end{equation} 
The same argument applies here and we have $H=U(d)$. This means $\{u\otimes\cdots\otimes u\otimes\1|u\in U(d)\}\subset G$. By recursive application of the argument we get $\{u\otimes\cdots\otimes u\otimes\1\otimes \1|u\in U(d)\}\subset G$ all the way until we get $U(d)\otimes\1\cdots\otimes \1\subset G$. Through permutations of the qudits, we also obtain that $\1\otimes\cdots\otimes\1\otimes U(d)\otimes\1\otimes\cdots\otimes\1\subset G$. If we take the product of all those sets, we finaly have that $U(d)^{\otimes N}\subset G$. Thus $U\in \mathcal{C}_N(U(d))$ and $\mathcal{C}_N(\mathcal{A})\subset \mathcal{C}_N(U(d))$.
 \end{proof}

\section{Proof of proposition~\ref{pro: non entangling}}
\label{annex: proof pro 2}
\begin{proof}[\underline{Proof}]
 	For the $i$-th qudit, we introduce an orthonormal basis $(\ket{\varphi_{i,j}})_{j}$ with $\ket{\varphi_{i,1}}=\ket{\varphi_i}$.\\
 	We then define the following operator
 	\begin{equation}
 		O=\bigotimes\limits_{k=1}^N \left(\sum\limits_{l=1}^d e^{i (l-1)\pi/d^k}\ketbra{\varphi_{k,l}} 		
 		\right).
 		\end{equation}
	We have $O\in U(d)^{\otimes N}$. By expanding the tensor product, one can verify that all the eigenvalues of $O$ are different. Furthermore if we consider an operator $A$ with eigenvalue $\lambda$ and corresponding eigenvector $\ket{\rho}$ we see that:
	\begin{equation}
		A\ket{\rho}=\lambda\ket{\rho} \Rightarrow UAU^\dagger (U\ket{\rho})=\lambda U\ket{\rho}.
	\end{equation}
Thus eigenvectors of $A$ are sent to eigenvectors of $UAU^\dagger$ with same eigenvalues. By hypothesis on $U$, we have $UOU^\dagger\in U(d)^{\otimes N}$, which means that $UOU^\dagger$ has an eigenbasis composed of product-states. Considering the uniqueness of the eigenvalues, it follows that $U\ket{\psi}$ is a product-state.
\end{proof}

\section{Two lemmas}
\label{appendix: lemma}
In this section, we present and prove two technical lemmas that will be utilized in the proof of Proposition~\ref{pro: non entangling are trivial}.
\begin{lem}\label{lemma: simple}
If the following conditions hold for $u,v,w,x\in\C^k\setminus\{0\}$:
\begin{itemize}
\item $u$ and $v$ are not collinear,
\item $w=\alpha u$ and $x=\beta v$,
\item $u+v\propto w+x$.
\end{itemize}
then $\alpha=\beta$.
\end{lem}

\begin{proof}[\underline{Proof}]
	 We represent vectors $w$ and $x$ in basis $\{u,v\}$, the condition $u+v\propto w+x$ becomes:
	\begin{equation}
		\left(\begin{matrix}
		1\\1
		\end{matrix}\right)\propto \left(\begin{matrix}
		\alpha\\ \beta
		\end{matrix}\right),
	\end{equation}
	thus $\alpha=\beta$.
\end{proof}

\begin{lem}\label{lemma: hard}
 Let $U$ be an invertible matrix of size $d^n$. There exist $d\times d$ unitary matrices $u_1,\dots, u_n$ and $v_1,\dots,v_n$ such that all the coefficients of the matrix 
\begin{equation}
	V=(u_1\otimes\cdots\otimes u_n)U(v_1\otimes \cdots\otimes v_n)
\end{equation}
 are all non-zero.
\end{lem}

\begin{proof}[\underline{Proof}]
The matrices $u_i$ and $v_j$ will be constructed with the help of a $d\times d$ matrix $M$ which depends on a real parameter $t$,
\begin{equation}
	M=(i+1+t^2+\cdots+t^{2(d-1)})\1-2\left(
	  \begin{matrix}
		1&t&t^2&\cdots&t^{d-1}\\
		t&t^2&t^3&\cdots&t^d\\
		\vdots&\vdots&\vdots&&\vdots\\
		t^{d-1}&t^d&t^{d+1}&\cdots&t^{2(d-1)}
	\end{matrix}\right).
\end{equation}
The general matrix element of $M$ is
\begin{equation}
	M_{k,l}=-2t^{k+l-2} + \delta_{kl}(i+\sum\limits_{s=0}^{d-1} t^{2s}).
\end{equation}
Firstly, let us check that the columns of $M$ are orthogonal to each other, {\it i.e.}, for $l_0\neq l_1$ we have

\begin{align}
	\sum\limits_{k=1}^d M_{k,l_0} M_{k,l_1}^\ast&=\sum\limits_{k=1}^d \left(-2t^{k+l_0-2}+\delta_{kl_0}i+\delta_{k,l_0}\sum\limits_{s=0}^{d-1}t^{s}\right)\left(-2t^{k+l_1-2}-\delta_{kl_1}i+\delta_{k,l_1}\sum\limits_{s=0}^{d-1}t^{s}\right)\\
	&=\sum\limits_{k=1}^d\Bigg[4t^{2k+l_0+l_1-4}+2t^{k+l_0-2}\delta_{kl_1}i-2t^{k+l_0-2}\delta_{kl_1}\sum\limits_{s=0}^{d-1} t^{2s}-2t^{k+l_1-2}\delta_{kl_0}i-2t^{k+l_1-2}\delta_{kl_0}\sum\limits_{s=0}^{d-1}t^{2s}\Bigg]\\
	&=4t^{l_0+l_1-2}\sum\limits_{k=0}^{d-1}t^{2k}+\cancel{2t^{l_0+l_1-2}i}-\cancel{2t^{l_0+l_1-2}i}-2t^{l_0+l_1-2}\sum\limits_{s=0}^{d-1}t^{2s}-2t^{l_0+l_1-2}\sum\limits_{s=0}^{d-1}t^{2s}\\
	&=0.
\end{align}
Then we verify that all the columns of matrix $M$ have the same norm. So for a fixed $l$ we compute
\begin{align}
	\sum\limits_{k=1}^d M_{k,l} M_{k,l}^\ast&=\sum\limits_{k=1}^d \left(-2t^{k+l-2}+\delta_{kl}i+\delta_{k,l}\sum\limits_{s=0}^{d-1}t^{s}\right)\left(-2t^{k+l-2}-\delta_{kl}i+\delta_{k,l}\sum\limits_{s=0}^{d-1}t^{s}\right)\\
	&=\sum\limits_{k=1}^d\Bigg[4t^{2k+2l-4}-2t^{k+l-2}\delta_{kl}\sum\limits_{s=0}^{d-1} t^{2s}+\delta_{kl}-2t^{k+l-2}\delta_{kl}\sum\limits_{s=0}^{d-1}t^{2s} +\delta_{kl}\left(\sum\limits_{s=0}^{d-1}t^{2s}\right)^2\Bigg]\\
	&=4t^{2l-2}\sum\limits_{k=0}^{d-1}t^{2k}+1-4t^{2l-2}\sum\limits_{s=0}^{d-1} t^{2s}+\left(\sum\limits_{s=0}^{d-1}t^{2s}\right)^2\\
	&=1+\left(\sum\limits_{s=0}^{d-1} t^{2s}\right)^2\neq 0.
\end{align}
We see that the last expression is independent of the column index $l$. The last two computations show that for all $t$ the matrix $M$ is proportional to a unitary matrix. This means that we can take $\tilde{u_i}=M(t_i)$ and $\tilde{v_j}=M(t_j')$ and show that there exist some values $t_i$ and $t_j'$ such that the requirements of the lemma are met. At the final step, we simply normalize the matrices $\tilde u_i$ and $\tilde v_j$. To prove the existence of $t_i$ and $t'_j$, we show the following statement by induction:\\

For all finite family $\{A^{(i)}\}_{i\in I}$ of non-zero $d^n\times d^n$ matrices ({\it i.e.}, $A^{(i)}\in \mathcal M_{d^n}(\C)\setminus \{0\}$), there exist some real numbers $t_1,\dots,t_n$ and $t_1',\dots,t_n'$ such that
\begin{equation}
    \forall i\in I,~\Big(M(t_1)\otimes \cdots\otimes M(t_n)\Big)A^{(i)}\Big(M(t_1')\otimes \cdots \otimes M(t_n')\Big)\in \mathcal M_{d^n}(\C\setminus\{0\}),
\end{equation}
{\it i.e.}, are matrices with only non zero coefficients.\\

If $n=0$, then the $A^{(i)}$'s are non zero $1\times 1$ matrix, {\it i.e.}, single non zero numbers. In this case there is nothing to show.\\

If $n\geq 1$, lets consider a finite set $\{A^{(i)}\}_{i\in I} $ of non zero matrices. We  write each of these matrices in block form as

\begin{equation}\label{eq: diff block form}
		A^{(i)}=\left(\begin{matrix}
				A^{(i)}_{1,1}&\cdots&A^{(i)}_{1,d}\\
				\vdots&&\vdots\\
				A^{(i)}_{d,1}&\cdots&A^{(i)}_{d,d}
			\end{matrix}\right),
	\end{equation}
where the $A^{(i)}_{k,l}$'s are $d^{n-1}\times d^{n-1}$ square matrices. We shall firstly look at the product $(M(t)\otimes\1\otimes \cdots \otimes \1)A^{(i)}$:

\begin{align}
	\label{equmU}
	\Big[(M(t)\otimes\1\otimes \cdots \otimes \1)A^{(i)}\Big]_{k,l}&=t^{k-1}(A^{(i)}_{1,l}+tA^{(i)}_{2,l}+t^2A^{(i)}_{3,l}+\cdots+t^{d-1}A^{(i)}_{d,l})-2A^{(i)}_{k,l}(i+1+t^2+\cdots+t^{2(d-1)})\\
	&:=P^{(i)}_{k,l}(t),
\end{align}
calculated at the square block at position $(k,l)$. Next, we multiply expression (\ref{equmU}) by $M(t')\otimes\1\otimes \cdots\otimes \1$ from the right-hand side with the choice $t'=t^{2d}$ and we get while considering the block at position $(k,l)$
\begin{align}
    \label{equmUn}
    \Big [(M(t)&\otimes\1\otimes \cdots \otimes \1)A^{(i)}(M(t^{2d})\otimes\1\otimes \cdots\otimes \1)\Big]_{k,l}\notag\\
    &= (t^{2d})^{l-1}\left[P^{(i)}_{k,1}(t)+t^{2d}P^{(i)}_{k,2}(t)+(t^{2d})^2P^{(i)}_{k,3}(t)+\cdots+(t^{2d})^{d-1}P^{(i)}_{k,d}(t)\right]\\
    &~~~~~~-2P^{(i)}_{k,l}(t)(i+1+(t^{2d})^2+\cdots+(t^{2d})^{2(d-1)}).\notag\\
    &:=Q^{(i)}_{k,l}(t)
\end{align}

This complicated expression is a polynomial function in $t$ with matrix coefficients. By this, we mean a function of the form $\sum_j B_j t^j$ where the coefficients $B_j$ are matrices. We can verify that the polynomial of equation (\ref{equmUn}) is non zero, {\it i.e.}, at least one of the matrix coefficient is non-zero. To see why this is the case, lets assume by contradiction that all coefficients of eq.~(\ref{equmUn}) are zero. Since the degree of all $P^{(i)}_{\alpha,\beta}(t)$ ($1\leq \alpha\leq d$ and $1\leq \beta\leq d$) is less than $2d$, by looking at the term of degree less that $2d$, it would imply that all the coefficients of $P^{(i)}_{k,l}(t)$ are zero. And then that the coefficients of all the $P^{(i)}_{k,\beta}(t)$ are also zero for $1\leq \beta\leq d$. Finally, by considering eq.~(\ref{equmU}), it would imply that all the $A^{(i)}_{\alpha,\beta}$ are also zero. This is clearly a contradiction, since we assumed that the matrix $A^{(i)}$ is non-zero. So for every $k$ and $l$, the polynomial with matrix coefficient $Q^{(i)}_{k,l}(t)$, is non zero. Such polynomials have at most a finite number of roots, and since the set $I$ is also finite, there must exist a value $t_1$ such that 
\begin{equation}
    \forall i\in I,~\forall 1\leq k\leq d,~\forall 1\leq l\leq d,~Q_{k,l}^{(i)}(t_1)\neq 0,
\end{equation}
{\it i.e.}, are not the zero matrix. We can thus write, with $t_1'=t_1^{2d}$:
\begin{equation}
    (M(t_1)\otimes\1\otimes \cdots \otimes \1)A^{(i)}(M(t_1')\otimes\1\otimes \cdots\otimes \1)=\left(\begin{matrix}
				Q^{(i)}_{1,1}(t_1)&\cdots&Q^{(i)}_{1,d}(t_1)\\
				\vdots&&\vdots\\
				Q^{(i)}_{d,1}(t_1)&\cdots&Q^{(i)}_{d,d}(t_1)
			\end{matrix}\right).
\end{equation}
Applying the induction hypothesis to the set $\Big\{Q^{(i)}_{j,k}(t_1)\Big|i\in I,~1\leq j\leq d,~1\leq k\leq d\Big\}$, which is a finite set of $d^{n-1}\times d^{n-1}$ non-zero matrices, we know that there exist $2(n-1)$ real parameters $t_2,\dots,t_n$ and $t_2',\dots,t_n'$ such that 
\begin{equation}
    \Big(M(t_2)\otimes \cdots\otimes M(t_n)\Big)Q^{(i)}_{j,k}(t_1)\Big(M(t_2')\otimes \cdots \otimes M(t_n')\Big):=\tilde{A}^{(i)}_{k,l}
\end{equation}
are matrices with only non-zero coefficients, for all values of $i$, $k$ and $l$. Since we have 
\begin{equation}
    \Big(\1\otimes M(t_2)\otimes \cdots\otimes M(t_n)\Big)\left(\begin{matrix}
				Q^{(i)}_{1,1}(t_1)&\cdots&Q^{(i)}_{1,d}(t_1)\\
				\vdots&&\vdots\\
				Q^{(i)}_{d,1}(t_1)&\cdots&Q^{(i)}_{d,d}(t_1)
			\end{matrix}\right)\Big(\1\otimes M(t_2')\otimes \cdots \otimes M(t_n')\Big)=\left(\begin{matrix}
				\tilde{A}^{(i)}_{1,1}&\cdots&\tilde{A}^{(i)}_{1,d}\\
				\vdots&&\vdots\\
				\tilde{A}^{(i)}_{d,1}&\cdots&\tilde{A}^{(i)}_{d,d}
			\end{matrix}\right),
\end{equation}
we finally have for each $i$,
\begin{equation}
    \Big(M(t_1)\otimes \cdots\otimes M(t_n)\Big)A^{(i)}\Big(M(t_1')\otimes \cdots \otimes M(t_n')\Big)=\left(\begin{matrix}
				\tilde{A}^{(i)}_{1,1}&\cdots&\tilde{A}^{(i)}_{1,d}\\
				\vdots&&\vdots\\
				\tilde{A}^{(i)}_{d,1}&\cdots&\tilde{A}^{(i)}_{d,d}
			\end{matrix}\right),
\end{equation}
which are indeed a matrices with only non-zero coefficients. This conclude the proof by induction. \\

The proof of the Lemma $\ref{lemma: hard}$ follow from the result shown by induction, by applying it to the set $\{U\}$ containing the single element $U$ (which is indeed a non-zero matrix, because it is invertible). This tells us that there exist some real number $t_1,\dots,t_n$ and $t'_1,\dots,t'_n$ such that 

\begin{equation}
    \Big(M(t_1)\otimes \cdots\otimes M(t_n)\Big)U\Big(M(t_1')\otimes \cdots \otimes M(t_n')\Big)
\end{equation}
is a matrix with only non-zero coefficients. Defining the unitary matrices
\begin{align}
    u_i=M(t_i)\Bigg/\sqrt{1+\left(\sum\limits_{s=0}^{d-1}t_i^{2s}\right)^2} && v_i=M(t_i')\Bigg/\sqrt{1+\left(\sum\limits_{s=0}^{d-1}t_i'^{2s}\right)^2}
\end{align}
concludes the proof.
\end{proof}

\section{Proof of proposition~\ref{pro: non entangling are trivial}}
\label{annex: proof prop 3}

\begin{proof}[\underline{Proof}]
In this proof, we use the two lemmas of the Appendix~\ref{appendix: lemma}. To ensure clarity, we will first explain the case of qubits and then generalize it for qudits , {\it i.e.}, a system of general dimension $d$.\\

\underline{Case $d=2$ and $N=2$}\\
We begin with the case $N=2$ and then proceed by induction. Let's consider  $U\in\mathcal{C}_2(U(2))$. Based on the Lemma~\ref{lemma: hard}, we know that we can compose $U$ with local gates so that we can assume without loos of generality that the coefficient of $U$ are all non-zero. We can then write $U$ in block form with each block being non-zero:
\begin{equation}
    U=\left(\begin{matrix}x_1&x_2&x_3&x_4\\y_1&y_2&y_3&y_4\end{matrix}\right)=\left(\begin{matrix}z_1&t_1\\z_2&t_2\\z_3&t_3\\z_4&t_4\end{matrix}\right).
\end{equation}

Since $\ket{00}=\left(\begin{matrix} 1&0&0&0\end{matrix}\right)^T$ is a product-state, $U\ket{00}=\left(\begin{matrix} x_1\\y_1\end{matrix}\right)$ is also a product-state. This implies that the blocks $x_1$ and $y_1$ are proportional to each other, and similarly, $x_2\propto y_2$, $x_3\propto y_3$ and $x_4\propto y_4$. By the same procedure with  $\ket{0}\otimes(\ket{0}+\ket{1})=\left(\begin{matrix} 1&1&0&0\end{matrix}\right)^T$ and $(\ket{0}+\ket{1})\otimes\ket{0}=\left(\begin{matrix} 1&0&1&0\end{matrix}\right)^T$ which are also product-states, we observe that $x_1+x_2\propto y_1+y_2$ and $x_1+x_3\propto y_1+y_3$.\\
	
We have seen that $U$ preserve the product ``\textit{bra}"-states (see the discussion after prop.~\ref{pro: non entangling}). Taking the hermitian conjugate thus shows that $U^{-1}=U^{\dagger}$ also preserve the product-state structure. This implies that $U$ cannot map an entangled state to a product one, otherwise, $U^{-1}$ would send a product-state to an entangled one. So, $\ket{01}+\ket{10}=\left(\begin{matrix} 0&1&1&0\end{matrix}\right)^T$ which is not a product-state, is mapped to an entangled one. This implies $x_2+x_3\not\propto y_2+y_3$. Since all the vectors involved are non-zero we have  $x_2\not\propto x_3$ otherwise it would imply $x_2+x_3\propto y_2+y_3$. In complete analogy, we have $x_1\not\propto x_4$. We get the same type of equation by considering the expressions $\bra{\psi}U$, for analogously chosen $\bra{\psi}$: $z_1\propto t_1$, $z_2\propto t_2$, $z_3\propto t_3$, $z_4\propto t_4$, $z_1+z_2\propto t_1+t_2$... as well as $z_1\not\propto z_4$ and $z_2\not\propto z_3$.\\
	
Without loss of generality, we can assume that we have $x_1\not\propto x_2$ and $x_3\not\propto x_4$. Indeed the previously shown proportionality relations imply that we have either $(x_1\not\propto x_2~\&~x_4\not\propto x_3)$ or $(x_1\not\propto x_3~\&~x_4\not\propto x_2)$. So, if necessary, we can multiply $U$ by the swapping gate $P=\left(\begin{matrix}1&0&0&0\\0&0&1&0\\0&1&0&0\\0&0&0&1	\end{matrix}\right)$, whose effect is to exchange the second and third columns of $U$, to have $x_1\not\propto x_2$ and $x_3\not\propto x_4$. The Lemma~\ref{lemma: simple} then implies proportionality relation between larger blocks $\left(\begin{matrix}x_1&x_2\end{matrix}\right)\propto\left(\begin{matrix}y_1&y_2\end{matrix}\right)$ and $\left(\begin{matrix}x_3&x_4\end{matrix}\right)\propto\left(\begin{matrix}y_3&y_4\end{matrix}\right)$.\\
	
Now, we want to apply the same procedure to the rows. We repeat the same process as before, but this time we cannot use a swapping gate. This is not a problem since, thanks to the proportionality relations between larger blocks, we know that $z_1\propto z_3$. This implies $z_1\not\propto z_2$ and $z_3\not\propto z_4$, and we can apply the Lemma~\ref{lemma: simple} again to get $\left(\begin{matrix}z_1\\z_2\end{matrix}\right)\propto\left(\begin{matrix}t_1\\t_2\end{matrix}\right)$ and $\left(\begin{matrix}z_3\\z_4\end{matrix}\right)\propto\left(\begin{matrix}t_3\\t_4\end{matrix}\right)$.\\
	
The proportionality relation obtained for the bigger blocks implies that $U$ can be written in the form of a tensor product: $U=V\otimes W$. Furthermore, since $U$ is unitary, we have $UU^{\dagger}=\1=(VV^\dagger)\otimes (WW^\dagger)$. This implies that $VV^\dagger\propto WW^\dagger\propto \1$, indicating that with the appropriate normalization factor, $V$ and $W$ can be chosen to be unitary.  Since we used the possibility of composition with a swapping gate to reach this form, we have obtained the result of the theorem for the case $N=2$ and $d=2$.\\
	
\underline{Case $d=2$ and $N>2$}\\
The case with $N$ qubits is similar. We will show that $U$ can be factored on the first qubit and then the proof follows by induction. We apply the same procedure by writing $U$
	\begin{equation}
		U=\left(\begin{matrix}
				A&B\\
				C&D
			\end{matrix}\right)
		=\left(\begin{matrix}
				x_1&\cdots&x_{2^N}\\
				y_1&\cdots&y_{2^N}
			\end{matrix}\right)
			=\left(\begin{matrix}
				z_1&t_1\\
				\vdots&\vdots\\
				z_{2^N}&t_{2^N}
			\end{matrix}\right),
	\end{equation}
where $A$, $B$, $C$, $D$ are square matrices of size $2^{N-1}$, the $x_i$, $y_i$ are length $2^{N-1}$ column vectors and $z_i$ and $t_i$ are length $2^{N-1}$ row vectors. Thanks to Lemma~\ref{lemma: hard}, we can still assume without loss of generality all the coefficients of $U$ are non zero, and thus that $x_i$, $y_i$, $z_i$ and $t_i$ are all non-zero blocks.  

In the following discussion, instead of referring to each column or row with an index $i$ running from $1$ to $2^N$, we will label them using the binary string associated to the corresponding basis state $\ket{i_1,\cdots,i_N}$. For instance, we will speak of the column labeled as $000\cdots 0100$, which would be the fourth column. We start by examining the columns, more specifically only the upper half, \textit{i.e.}, the $x_i$'s. By using the same methodology as before, we know that the following holds:
\begin{enumerate}[label=\Alph*)]
    \item If the strings $i$ and $j$ labeling two columns differ at exactly one position and if $x_i\not \propto x_j$, then $\left(\begin{matrix}x_i&x_j\end{matrix}\right)\propto\left(\begin{matrix}y_i&y_j\end{matrix}\right)$. 
        
    Indeed under these hypothesis the sum of basis states $\ket{i}+\ket{j}$ is a product-state, which then implies that $x_i+x_j\propto y_i+y_j$. We can apply lemma~\ref{lemma: simple} which allows to conclude that $\left(\begin{matrix}x_i&x_j\end{matrix}\right)\propto\left(\begin{matrix}y_i&y_j\end{matrix}\right)$. \label{point: A}
    \item If the strings $i$ and $j$ of two columns differ at two positions or more, then we know that $x_i\not\propto x_j$.

    To see why this is true, suppose the strings differ at two positions. The sum $\ket{i}+\ket{j}$ is an entangled state and its image  $\left(\begin{matrix}x_i+x_j\\y_i+y_j\end{matrix}\right)$ is thus also entangled. Now assume by contradiction that $x_i\propto x_j$. Then $\left(\begin{matrix}x_i\\x_j\end{matrix}\right)$ is a product-state, which can be written as $\left(\begin{matrix}x_i\\x_j\end{matrix}\right)=\ket{\psi_1}\otimes\cdots\otimes\ket{\psi_N}$. Consequently, $x_i$ must be the vector $a\ket{\psi_2}\otimes\cdots\otimes\ket{\psi_N}$. All the proportionality relation then imply that $x_j=b\ket{\psi_2}\otimes\cdots\otimes\ket{\psi_N}$, $y_i=c\ket{\psi_2}\otimes\cdots\otimes\ket{\psi_N}$ and $y_j=d\ket{\psi_2}\otimes\cdots\otimes\ket{\psi_N}$. Meaning that $\left(\begin{matrix}x_i+x_j\\y_i+y_j\end{matrix}\right)=(a+b,c+d)^T\otimes \ket{\psi_2}\otimes\cdots\otimes\ket{\psi_N}$, which is a contradiction since it is supposed to be an entangled state. \label{point: B}
\end{enumerate}

Now we begin the analysis with  $x_{0\cdots 0}$. Using the point~\ref{point: B} from above we know that,  $x_{0\cdots 01}\not\propto x_{10\cdots 0}$ and 
 thus $x_{0\cdots 0}$ has to be non-collinear to either $x_{0\cdots 01}$ or $x_{10\cdots 0}$. If necessary, we can multiply $U$ by a swapping gate that exchanges qubit $1$ and qubit $N$. This permutes the columns $0\cdots 01$ and $10\cdots 0$ so that, we can assume that $x_{0\cdots 0}\not\propto x_{0\cdots 01}$. Similarly, we can assume $x_{0\cdots 0}\not\propto x_{0\cdots 010}$ by using the same trick of multiply $U$ by the swapping gate that exchange qubit $1$ and qubit $N-1$. This would not affect the relation $x_{0\cdots 0}\not\propto x_{0\cdots 01}$ we just shown since the corresponding columns are invariant under this exchange of qubits. The argument continues in the same fashion: we have $x_{0\cdots 0}\not\propto x_i$, where the string $i$ has exactly one `$1$' that is not at the first position. The point~\ref{point: A} tells us that the proportionality coefficients $\lambda_i$ (such that  $x_i=\lambda_i y_i$) are all the same for $i=0\cdots 0$ or $i$ has exactly one  `$1$' that is not at the first position. We will continue with the same methodology propagating the fact that $\lambda_i$ is a constant for all string $i$ starting with `$0$'. Indeed, $x_{0\cdots 01}\not\propto x_{0\cdots 010}$ thus $x_{0\cdots 011}$ is non-proportional to at least $x_{0\cdots 01}$ or $x_{0\cdots 010}$ and we can apply the point~\ref{point: A}, to get $\lambda_{0\cdots 011}=\lambda_{0\cdots 0}$. This extension applies to all $x_i$ with $i$ having two `$1$' (not on the first qubit), then three, and so on, until $x_{01\cdots 1}$. In the end, we get that $A\propto C$.\\
	
We now apply the same reasoning for rows of $U$, except we do not have the flexibility of multiplying by some swapping gates. Nonetheless, this is not needed, because $z_{0\cdots 0}$ is non-proportional to at least $z_{0\cdots 01}$ or $z_{10\cdots 0}$ and we have just shown that $z_{0\cdots 0}\propto z_{10\cdots 0}$. Thus $A\propto B$. The same thing works for the lower half: $C\propto D$. The proportionality relation between the large block $A$, $B$, $C$ and $D$ finally show that $B\propto D$.\\
	
All these four proportionality relations between large blocks ($A\propto B$, $B\propto D$, $A\propto C$ and $C\propto D$) imply that $U$ can be written in the form of a tensor product: $U=V\otimes W$. Here $V$ is a $2\times 2$ matrix, and $W$ is $2^{N-1}\times 2^{N-1}$. Furthermore, $U$ is unitary, which means that we can chose $V$ and $W$ to be unitary as well. $W$ will also stabilize the set of product-states on $N-1$ qubit, thus we can apply the induction hypothesis to $W$. We used the possibility to compose $U$ with swap gates to arrive at this form but since for a local gate $u$ and a permutation gate $P$, $PuP^\dagger$ is local, we can ``move" the permutation gate to the left to arrive at the result of the theorem in the case $d=2$.\\
	  
\underline{Case $d\geq 1$ and $N\geq 1$}\\
For the case of $N$ qudits, the proof works in a similar way. We start by applying the Lemma~\ref{lemma: hard} and multiply $U$ by local gates, so that we  write
\begin{equation}\label{eq: decomp U general}
    U=\left(\begin{matrix}
    A_{1,1}&\cdots&A_{1,d}\\
    \vdots&&\vdots\\
    A_{d,1}&\cdots&A_{d,d}
    \end{matrix}\right)=\left(\begin{matrix}
      \alpha_{1,1}&\cdots&\alpha_{1,d^N}\\
    \vdots&&\vdots\\
    \alpha_{d,1}&\cdots&\alpha_{d,d^N}
    \end{matrix}\right),
\end{equation}
with all $\alpha$-blocks being non-zero column vectors of length $d$. Acting on the right of equation (\ref{eq: decomp U general}) by a basis state, it becomes clear that for all fixed $j$, $\alpha_{1,j},\alpha_{2,j},\dots, \alpha_{d,j}$ are proportional to each other. Let's denote the proportionality coefficients as $\lambda_{i,j}$, {\it i.e.}, $\alpha_{i,j}=\lambda_{i,j}\alpha_{1,j}$. Now, what we have to show is that $\lambda_{i,j}=\lambda_{i,0\cdots0}$ for all $i=1,\dots ,d$ and string $j$ which start with a `$0$'.\\

By using remarks~\ref{point: A} and~\ref{point: B} listed previously and generalized on qudits (replace $x_j$ by $\alpha_{1,j}$), we know that $\lambda_{i,j}=\lambda_{i,0\cdots0}$, for either  $j=0\cdots 01$ or $j=10\cdots 0$. As before, we can multiply $U$ by the swapping gate permuting qudit $1$ and $N$, such that we get $\lambda_{i,0\cdots01}=\lambda_{i,0\cdots0}$. Like in the qubit case we show that the swapping freedom allows us to assume $\alpha_{1,0\cdots0}\not\propto \alpha_{1,j}$ whenever $j$ is a string composed of only zeros and one $1$ which is not at the beginning. For such string $j$, we have that $\lambda_{i,j}=\lambda_{i,0\cdots0}$. \\

For the column $0\cdots 02$, we know $\alpha_{1,0\cdots 02}$ cannot be proportional to  $\alpha_{1,0\cdots 0}$ and $\alpha_{1,0\cdots 01}$  at the same time (since both are non-collinear). Thus by point~\ref{point: A}, we have $\lambda_{i,0\cdots 02}=\lambda_{i,0\cdots 0}$ or $\lambda_{i,0\cdots 01}$. But since we already showed $\lambda_{i,0\cdots 0}=\lambda_{i,0\cdots 01}$ we indeed manage to obtain $\lambda_{i,0\cdots 02}=\lambda_{i,0\cdots 0}$. The same holds ($\lambda_{i,j}=\lambda_{i,0\cdots 0}$) for all string $j$ only made of `$0$' and one non-zero index which is not at the first place.\\

We can then gradually extend this equality of the proportionality coefficients ($\lambda_{i,j}=\lambda_{i,0\cdots 0}$) to all columns with their label $j$ starting with `$0$'. In this case, we reason like for the qubit case: we compare the column $j$ we are interested in, with two other columns $k$ and $l$ for which we already know that $\lambda_{i, k} =\lambda_{i,l} =\lambda_{i,0\cdots0}$. We can choose the strings $l$ and $k$ such that: $k$ differ from $j$ at one position, $l$ differ form $j$ at one position and $k$ differ form $l$ at two positions. Application of point~\ref{point: B} then shows that $\alpha_{1,k}\not\propto\alpha_{1,l}$ and thus $\alpha_{1,j}$ is not proportional to $\alpha_{1,k}$ or $\alpha_{1,l}$. The point~\ref{point: A} allows then to conclude that $\lambda_{i,j} =\lambda_{i,0\cdots 0}$.\\

The last equality means that the $d^{N-1}\times d^{N-1}$ block $A_{i,1}$ ( $i=1,\dots, d$) are all proportional to each other: $A_{1,1}\propto A_{2,1}\propto\cdots\propto A_{d,1}$. We now continue the analysis on the rows. This time multiplying $U$ by swapping gates is not necessary, like in the case of qubits, because we already know that the row $0\cdots 0$ and $10\cdots 0$  are proportional. Thus, this procedures show that all the square blocks on the same line are proportional to each other, that is, for every $i$, $A_{i,1}\propto \cdots\propto A_{i,d}$. This, with the proportionality relation we showed for the $A_{i,1}$'s block, shows that in fact all the square blocks are proportional to each other, regardless of their positions. This means $U=V\otimes W$ where $V$ (of size $d\times d$) and $W$ (of size $d^{N-1}\times d^{N-1}$) can be chosen to be unitary. Additionally, since $W$ stabilizes product-states on $N-1$ qudits, we can conclude the proof by induction.

\end{proof}

\section{Proof of theorem~\ref{thm: eigenvalue of projectors}}
\label{annex: proof thm 6}
\begin{proof}[\underline{Proof}]
    Since the projector $P_{0,\cdots,0}$ is of rank $2^{n-1}$, we will have at least $2^{n-1}$ zeros as eigenvalues. To show that the states defined by the equation (\ref{equstate}) are indeed eigenvectors with their respective eigenvalues, we proceed with the following computation
    \begingroup
    \allowdisplaybreaks
    \begin{subequations}
        \begin{align}	
			P_{\theta_1,\cdots,\theta_n}\ket{\psi_{j_1,\cdots,j_n}}&=\frac{1}{2}\ket{\psi_{j_1,\cdots,j_n}}+\frac{\operatorname{Re}}{2}\Bigg[\sum\limits_{k_1,\cdots,k_n} (i j_1)^{k_1}\cdots(i j_n)^{k_n}A_{\theta_1}\ket{k_1}\cdots A_{\theta_n}\ket{k_n}\Bigg]\\
			&=\frac{1}{2}\ket{\psi_{j_1,\cdots,j_n}}+\frac{\operatorname{Re}}{2} \Bigg[\prod\limits_{l=1}^n(\cos\theta_l\ket{0}+\sin\theta_l\ket{1}-ij_l\cos\theta_l\ket{1}+ij_l\sin\theta_l\ket{0}\Bigg]\\
			&=\frac{1}{2}\ket{\psi_{j_1,\cdots,j_n}}+\frac{\operatorname{Re}}{2} \Bigg[\prod\limits_{l=1}^n(e^{ij_l\theta_l}\ket{0}-ij_l e^{ij_l\theta_l}\ket{1})\Bigg]\\
			&=\frac{1}{2}\ket{\psi_{j_1,\cdots,j_n}}+\frac{\operatorname{Re}}{2} \Bigg[e^{i(j_1\theta_1+\cdots+j_n\theta_n)}\sum\limits_{k_1,\cdots,k_n}(-ij_1)^{k_1}\cdots (-ij_n)^{k_n}\ket{k_1\cdots k_n}\Bigg]\\
			&=\frac{1}{2}\ket{\psi_{j_1,\cdots,j_n}}+\frac{1}{4} \Bigg[e^{i(j_1\theta_1+\cdots+j_n\theta_n)}\sum\limits_{k_1,\cdots,k_n}(-ij_1)^{k_1}\cdots (-ij_n)^{k_n}\ket{k_1\cdots k_n}\\
			&~~~~~~+e^{-i(j_1\theta_1+\cdots+j_n\theta_n)}\sum\limits_{k_1,\cdots,k_n}(ij_1)^{k_1}\cdots (ij_n)^{k_n}\ket{k_1\cdots k_n}\Bigg]\\
			&=\frac{1}{2}\ket{\psi_{j_1,\cdots,j_n}}+\frac{1}{4} \Big[\sum\limits_{k_1,\cdots,k_n=0,1}\Big(e^{i(j_1\theta_1+\cdots+j_n\theta_n)}(-ij_1)^{k_1}\cdots (-ij_n)^{k_n}\\
			&~~~~~~~~+e^{-i(j_1\theta_1+\cdots+j_n\theta_n)}(ij_1)^{k_1}\cdots (ij_n)^{k_n}\Big)\ket{k_1\cdots k_n}\Bigg].
        \end{align}
    \end{subequations}
\endgroup
	
We can examine the terms within the last summation symbol and divide our analysis into two cases based on the parity of the number of indices $k_l$ that have the value $1$.

\begin{itemize}
    \item For an even number of $k_l$'s taking the value $1$, the term in the sum is $ (ij_1)^{k_1}\cdots(ij_n)^{k_n} 2 \cos(j_1\theta_1+\cdots+j_n\theta_n)\ket{k_1\cdots k_n}$.  Summing over all such terms yields $2\cos(j_1\theta1+\cdots+j_n\theta_n)\ket{\psi_{j_1,\cdots,j_n}}$.
	\item For an odd number of $k_l$'s taking the value $1$, the term in the sum becomes $ (-ij_1)^{k_1}\cdots(-ij_n)^{k_n} 2i \sin(j_1\theta_1+\cdots+j_n\theta_n)\ket{k_1\cdots k_n}$. These terms are all eigenvectors of $\sigma_Z\otimes\cdots\otimes \sigma_Z$ and we denote their sum as $\ket{\varphi}$.
\end{itemize}
Therefore, we have
\begin{equation}\label{eq: one proj on state}
P_{\theta_1,\cdots,\theta_n}\ket{\psi_{j_1,\cdots,j_n}}=\frac{1}{2}\ket{\psi_{j_1,\cdots,j_n}}+\frac{\cos(j_1\theta_1+\cdots+j_n\theta_n)}{2}\ket{\psi_{j_1,\cdots,j_n}}+\ket{\varphi}.
\end{equation}
Since $\ket{\varphi}$ is a $-1$-eigenvector of $\sigma_Z\otimes\cdots\otimes \sigma_Z$, by applying $P_{0,\cdots,0}$ on the left-hand side of equation (\ref{eq: one proj on state}) will yield zero. Finally, we conclude that
\begin{equation}
    P_{0,\cdots,0}P_{\theta_1,\cdots,\theta_n}\ket{\psi_{j_1,\cdots,j_n}}=\frac{1+\cos(j_1\theta1+\cdots+j_n\theta_n)}{2}\ket{\psi_{j_1,\cdots,j_n}}.
	\end{equation}
\end{proof}

\section{3 qubit stabilization pattern analysis} \label{annex: 3 qubit analysis}
In this section of the appendix, we describe in details the analysis of the stabilization involving
$3$ qubits. Following the definitions of~\ref{section: methods}, we seek minimal and unique stabilization. To achieve this, we will need either $2$ or $3$ stabilizers as defined in equation (\ref{eq: def list stab}). A operator, denoted as $\mathcal{O} = A_1 \otimes \cdots \otimes A_n$, where $A_i \in \{A_{\theta,\phi}, \1_2\}$, cannot have a unique eigenvalue of +1. Consequently, a single operator alone cannot uniquely stabilize a state. On the other hand, the investigation of stabilization with two operators $\mathcal O_1$ and $\mathcal O_2$ reveals that stabilization involving four or more operators is never minimal. Therefore, in the following two subsections, we examine the cases of stabilization using either $2$ or $3$ operators.

\subsection{Stabilization with $2$ operators}
All non-equivalent stabilization patterns with $2$ operators on $3$ qubits are as follows. We analyze these patterns using the methods described in section~\ref{section: methods}. Our focus extends beyond the states stabilized by these operators; we also examine the dimensions of the corresponding $+1$ eigenspaces, as a function of the various parameters. This leads to a two-operators compatibility condition: if multiple operators stabilize a state, then, for every pair of operators, the shared $+1$ eigenspace must have a minimum dimension of $1$. The results are presented in the table~\ref{tab: two op}.\\

\begin{table}[h]
\centering
\begingroup
\setlength{\tabcolsep}{6pt} 
\renewcommand{\arraystretch}{1.42} 

\begin{tabular}{|c|c|c|c|c|}
    \hline 
    No. & Operators & Dimension & Condition & Stabilizer state \\ 
    \hline \hline
    \multirow{4}{*}{\hypertarget{tab 2 1}{$1$}} & \multirow{4}{*}{
        $\begin{array}{ll}
            \sigma_Z\otimes \sigma_Z \otimes \sigma_Z \\ 
            A_\theta\otimes A_\phi \otimes A_\omega
	\end{array}$
    } & $4$  & $\theta=\phi=\omega= 0~[2\pi]$ & Non unique \\  
     &  & $2$  & $(\theta=0 ~\&~ \phi=\omega )$ or $(\theta=\pi ~\&~ \phi=\pi-\omega )$  and permutations & Non unique \\
     & & $1$  & Only one of the equation  $\theta\pm\phi\pm\omega=0~[2\pi]$ & $\ket{\psi_{\pm 1,\pm 1,\pm 1}}$ \\ 
     & & $0$  & Other & None \\ 
    \hline

    \multirow{3}{*}{\hypertarget{tab 2 2}{$2$}} & \multirow{3}{*}{
        $\begin{array}{ll}
            \sigma_Z\otimes \sigma_Z \otimes \sigma_Z \\ 
            A_\theta\otimes A_\phi \otimes \1
	\end{array}$
    } & $4$  & $\theta=\phi= 0,\pi~[2\pi]$ & Non unique \\ 
      & & $2$  & $\theta=\pm\phi$ & Non unique \\
      & & $0$  & Other & None \\ 
    \hline
    
    \multirow{2}{*}{\hypertarget{tab 2 3}{$3$}} &
    \multirow{2}{*}{
        $\begin{array}{ll}
            \sigma_Z\otimes \sigma_Z \otimes \sigma_Z \\ 
            A_\theta\otimes \1 \otimes \1
	\end{array}$
    } & $2$  & $\theta= 0,\pi~[2\pi]$ & Non unique \\ 
     & & $0$  & Other & None \\ 
    \hline

    \multirow{3}{*}{\hypertarget{tab 2 4}{$4$}} &
    \multirow{3}{*}{
        $\begin{array}{ll}
            \sigma_Z\otimes \sigma_Z \otimes \1 \\ 
            A_\theta\otimes A_\phi \otimes \1
	\end{array}$
    } & $4$  & $\theta=\phi= 0,\pi~[2\pi]$ & Non unique \\ 
      & & $2$  & $\theta=\pm\phi$ &  Non unique \\
       & & $0$  & Other & None \\ 
    \hline

    \multirow{2}{*}{\hypertarget{tab 2 5}{$5$}} &
    \multirow{2}{*}{
        $\begin{array}{ll}
            \sigma_Z\otimes \sigma_Z \otimes \1 \\ 
            A_\theta\otimes \1 \otimes \sigma_Z
	\end{array}$
    } & $2$  & $\theta= 0,\pi~[2\pi]$ & Non unique \\ 
       & & $0$  & Other & None \\ 
    \hline

    \multirow{2}{*}{\hypertarget{tab 2 6}{$6$}} &
    \multirow{2}{*}{
        $\begin{array}{ll}
            \sigma_Z\otimes \sigma_Z \otimes \1 \\ 
            A_\theta\otimes \1 \otimes \1
	\end{array}$
    } & $2$  & $\theta= 0,\pi~[2\pi]$ & Non unique \\ 
     & & $0$  & Other & None \\ 
    \hline

    \hypertarget{tab 2 7}{$7$} & $\begin{array}{ll}
        \sigma_Z\otimes \sigma_Z \otimes \1 \\ 
        \1\otimes \1 \otimes \sigma_Z
    \end{array}$
     & $2$  & None & Non unique\\ 
    \hline

    \multirow{2}{*}{\hypertarget{tab 2 8}{$8$}} &
    \multirow{2}{*}{
        $\begin{array}{ll}
            \sigma_Z\otimes \1 \otimes \1 \\ 
            A_\theta\otimes \1 \otimes \1
	\end{array}$
    } & $4$  & $\theta= 0~[2\pi]$ & Non unique \\ 
     & & $0$  & Other & None \\ 
    \hline

    \hypertarget{tab 2 9}{$9$} & $\begin{array}{ll}
        \sigma_Z\otimes \1 \otimes \1 \\ 
        \1\otimes \sigma_Z \otimes \1
    \end{array}$
     & $2$  & None & Non unique\\ 
    \hline
\end{tabular}
\endgroup
\caption{ Different ways two binary operators, constructed as tensor product of $\1$ and $A_\theta$ can stabilize subspaces of various dimension. The first column give the expression of the operators, the second give the list  of possible dimensions the stabilized subspace can have. The third, shows constraints on the set of free parameters that define the stabilizing operators, see equation (\ref{eq: def opera A theta phi}). The last column shows if a state is uniquely stabilized or not.}
\label{tab: two op}
\end{table}

From such analysis, we can observe the following facts
\begin{itemize}
    \item The only way to have unique stabilization with two operators, is when no tensor factor of the stabilizers are the identity operator $\1_2$ (correspond to the case number \hyperlink{tab 2 1}{$1$}). The corresponding state is $\ket{\psi_{\pm 1,\pm 1,\pm 1}}$ which is a Pauli stabilizer state, as shown section~\ref{section: binary case}.
    \item In all situations when two parameters $\theta$ and $\varphi$ are needed to define the stabilizers (case number \hyperlink{tab 2 2}{$2$} and \hyperlink{tab 2 4}{$4$}), the condition so that there is a two dimensional stabilized subspace is  $\theta=\pm\phi~[2\pi]$ and $\theta,\phi\neq 0,\pi$. This gives rise to a first instance of the two-operators compatibility rule.
    \item In a similar way, when only one parameter $\theta$ is needed to define the stabilizer (situation number \hyperlink{tab 2 3}{$3$}, \hyperlink{tab 2 5}{$5$}, \hyperlink{tab 2 6}{$6$} and \hyperlink{tab 2 8}{$8$}), the condition so that there is a two dimensional stabilized subspace is $\theta=0,\pi$. This correspond to another version of the two-operators compatibility rule.
    \item There is two situations (number \hyperlink{tab 2 7}{$7$} and \hyperlink{tab 2 9}{$9$}) where the stabilized subspace is of dimension $2$ space of dimension and no parameters are needed to define the stabilizers.
    \item Since we see that there is no situations, where the dimension of the stabilized subspace is $3$, it follows that we only need up to $3$ operators to do a unique stabilization on $3$ qubits. 
\end{itemize}

\subsection{Stabilization with 3 operators}
All non-equivalent  stabilization patterns are as follows. Here we have already discarded the situations where it is easy to see that unique stabilization is not possible, for example, due to too many identity operators used to construct the stabilizing operators. An example such trivial situation, is when the operators are $\sigma_Z\otimes\sigma_Z\otimes \1$, $A_\theta\otimes A_\phi\otimes \1$ and $A_{\psi,\alpha}\otimes A_{\kappa,\beta}\otimes \1$. In this case the last operators are all the identity $\1$, and thus either, there are no states stabilized or the dimension of the stabilized subspace is at least $2$. This kind of situation cannot uniquely stabilize a state and can safely be ignored. 

\begin{enumerate}
	\item $\left\{\begin{array}{ll}
	\sigma_Z\otimes \sigma_Z \otimes \sigma_Z \\ 
	A_{\theta_1}\otimes A_{\theta_2} \otimes A_{\theta_3}\\
	A_{\phi_1,\alpha}\otimes A_{\phi_2,\beta}\otimes A_{\phi_3,\gamma}
	\end{array} 
	\right.$.  The stabilized states are equivalent to $\ket{0}(\ket{00}+\ket{11})$ or $\ket{\psi_{\pm1,\pm1,\pm1}}$ (See details below). \label{case: 1}
	
	\item $\left\{\begin{array}{ll}
	\sigma_Z\otimes \sigma_Z \otimes \sigma_Z \\ 
	A_{\theta}\otimes A_{\phi} \otimes A_{\omega}\\
	A_{\psi,\alpha}\otimes A_{\kappa,\beta}\otimes \1
	\end{array} 
	\right.$. The stabilized states are equivalent to $\ket{0}(\ket{00}+\ket{11})$ (see details below).\label{case: 2}
	
	\item $\left\{\begin{array}{ll}
	\sigma_Z\otimes \sigma_Z \otimes \sigma_Z \\ 
	A_{\theta}\otimes A_{\phi} \otimes A_{\omega}\\
	A_{\psi,\alpha}\otimes \1\otimes \1
	\end{array} 
	\right.$. The stabilization is non-unique.
	
	\item $\left\{\begin{array}{ll}
	\sigma_Z\otimes \sigma_Z \otimes \sigma_Z \\ 
	A_{\theta}\otimes A_{\phi} \otimes \1\\
	A_{\omega,\alpha}\otimes A_{\psi,\beta}\otimes \1
	\end{array} 
	\right.$. The stabilized states are equivalent to $\ket{0}(\ket{00}+\ket{11})$.
	
	\item $\left\{\begin{array}{ll}
	\sigma_Z\otimes \sigma_Z \otimes \sigma_Z \\ 
	A_{\theta}\otimes A_{\phi} \otimes \1\\
	A_{\omega,\alpha}\otimes \1\otimes A_{\psi}
	\end{array} 
	\right.$. The stabilized states are equivalent to $\ket{000}$ or $\ket{000}+\ket{011}+\ket{101}+\ket{110}$.
	
	\item $\left\{\begin{array}{ll}
	\sigma_Z\otimes \sigma_Z \otimes \sigma_Z \\ 
	A_{\theta}\otimes A_{\phi} \otimes \1\\
	A_{\omega,\alpha}\otimes \1\otimes \1
	\end{array} 
	\right.$. Here $\theta=\pm \phi$, $\omega=0,\pi$ and the stabilizes states are equivalent to $\ket{000}$.
	
	\item $\left\{\begin{array}{ll}
	\sigma_Z\otimes \sigma_Z \otimes \sigma_Z \\ 
	A_{\theta}\otimes A_{\phi} \otimes \1\\
	\1\otimes\1\otimes A_\omega
	\end{array} 
	\right.$. Here $\theta=\pm \phi$, $\omega=0,\pi$ and the stabilized states are equivalent to $\ket{0}(\ket{00}+\ket{11})$.
	
	\item $\left\{\begin{array}{ll}
	\sigma_Z\otimes \sigma_Z \otimes \sigma_Z \\ 
	A_{\theta}\otimes \1 \otimes \1\\
	A_{\phi,\alpha}\otimes \1\otimes \1
	\end{array} 
	\right.$. The stabilization is non-unique.
	
	\item $\left\{\begin{array}{ll}
	\sigma_Z\otimes \sigma_Z \otimes \sigma_Z \\ 
	A_{\theta}\otimes \1 \otimes \1\\
	\1\otimes A_\phi\otimes \1
	\end{array} 
	\right.$. Here $\theta,\phi=0,\pi$ and the stabilized states are equivalent to $\ket{000}$.
	
	\item $\left\{\begin{array}{ll}
	\sigma_Z\otimes \sigma_Z \otimes \1 \\ 
	A_{\theta}\otimes A_\phi \otimes \1\\
	A_{\omega,\alpha}\otimes A_{\psi,\beta}\otimes \1
	\end{array} 
	\right.$. The stabilization is non-unique.
	
	\item $\left\{\begin{array}{ll}
	\sigma_Z\otimes \sigma_Z \otimes \1 \\ 
	A_{\theta}\otimes A_\phi \otimes \1\\
	A_{\omega,\alpha}\otimes \1\otimes \sigma_Z
	\end{array} 
	\right.$. The stabilization is non-unique.
	
	\item $\left\{\begin{array}{ll}
	\sigma_Z\otimes \sigma_Z \otimes \1 \\ 
	A_{\theta}\otimes A_\phi \otimes \1\\
	A_{\omega,\alpha}\otimes \1\otimes \1
	\end{array} 
	\right.$. The stabilization is non-unique.
	
	\item $\left\{\begin{array}{ll}
	\sigma_Z\otimes \sigma_Z \otimes \1 \\ 
	A_{\theta}\otimes A_\phi \otimes \1\\
	\1\otimes \1\otimes \sigma_Z
	\end{array} 
	\right.$. Here $\theta=\pm\phi$ and the stabilized states are equivalent to $(\ket{00}+\ket{11})\ket{0}$.
	
	\item $\left\{\begin{array}{ll}
	\sigma_Z\otimes \sigma_Z \otimes \1 \\ 
	A_{\theta}\otimes \1\otimes \sigma_Z\\
	\1\otimes A_\phi\otimes A_\omega
	\end{array} 
	\right.$. The stabilization is non-unique.
	
	\item $\left\{\begin{array}{ll}
	\sigma_Z\otimes \sigma_Z \otimes \1 \\ 
	A_{\theta}\otimes \1\otimes \sigma_Z\\
	A_{\phi,\alpha}\otimes \1\otimes \1
	\end{array} 
	\right.$. Here $\theta,\phi=0,\pi$ and the stabilized states are equivalent to $\ket{000}$.
	
	\item $\left\{\begin{array}{ll}
	\sigma_Z\otimes \sigma_Z \otimes \1 \\ 
	A_{\theta}\otimes \1\otimes \sigma_Z\\
	\1\otimes \1\otimes A_\phi
	\end{array} 
	\right.$. Here $\theta,\phi=0,\pi$ and the stabilized states are equivalent to $\ket{000}$.
	
	\item $\left\{\begin{array}{ll}
	\sigma_Z\otimes \sigma_Z \otimes \1 \\ 
	A_{\theta}\otimes \1\otimes \1\\
	A_{\phi,\alpha}\otimes \1\otimes \1
	\end{array} 
	\right.$. The stabilization is non-unique.
	
	\item $\left\{\begin{array}{ll}
	\sigma_Z\otimes \sigma_Z \otimes \1 \\ 
	A_{\theta}\otimes \1\otimes \1\\
	\1\otimes A_\phi\otimes \1
	\end{array} 
	\right.$. The stabilization is non-unique.
	
	\item $\left\{\begin{array}{ll}
	\sigma_Z\otimes \sigma_Z \otimes \1 \\ 
	A_{\theta}\otimes \1\otimes \1\\
	\1\otimes \1\otimes \sigma_Z
	\end{array} 
	\right.$. Here $\theta=0,\pi$ and the stabilized states are equivalent to $\ket{\psi}=\ket{000}$.
	
	\item $\left\{\begin{array}{ll}
	\sigma_Z\otimes \sigma_Z \otimes \1 \\ 
	\1\otimes \1\otimes \sigma_Z\\
	\1\otimes \1\otimes A_\theta
	\end{array} 
	\right.$. The stabilization is non-unique.
	
	\item $\left\{\begin{array}{ll}
	\sigma_Z\otimes \1 \otimes \1 \\ 
	\1\otimes \sigma_Z\otimes \1\\
	\1\otimes \1\otimes \sigma_Z
	\end{array} 
	\right.$. The stabilized states are equivalent to $\ket{\psi}=\ket{000}$.
\end{enumerate}

Since certain cases cannot be obtained by directly applying the methods described in Section~\ref{section: methods}, we will now elaborate on some of the more complex sub-cases.\\

For the case~\ref{case: 1}, we can use the study of stabilization with two operators, and more specifically the case \hyperlink{tab 2 1}{$1$} of table~\ref{tab: two op}. By applying the obtained results to the different pairs of operators we get that  $\theta_1=0$ and $\theta_2=\theta_3$ or $\theta_1=\pi$ and $\theta_2=\pi-\theta_3$, or any permutations thereof (the same holds for $\phi_1$, $\phi_2$, $\phi_3$). We can thus divide the problem into different sub-cases and use a numerical solver to get the complete set of solutions, which are set forth in the table~\ref{tab}. \\

\begin{table}[H]
\centering
\begingroup

        \setlength{\tabcolsep}{6pt} 
        \renewcommand{\arraystretch}{1.5} 
\begin{tabular}{|c|c|c|c|}
\hline 
\multicolumn{2}{|c|}{Operators} & Conditions & stabilizer state \\ 
\hline 
\multirow{4}{*}{$\sigma_Z\otimes A_{\varphi,0}\otimes A_{\varphi,0}$} &$\sigma_Z\otimes A_{\omega,\beta}\otimes A_{\omega,\beta} $  & \multirow{6}{*}{$\varphi,\omega,\beta \neq 0,\pi$} & $\ket{101}-\ket{110}$ \\ 
\cline{2-2}\cline{4-4}
&$\sigma_Z\otimes A_{\omega,\beta}\otimes A_{\omega,-\beta} $  &  & $\ket{000}+\ket{011}$ \\

\cline{2-2}\cline{4-4}
& $-\sigma_Z\otimes A_{\omega,\beta}\otimes A_{\omega,\pi-\beta} $  & & $\ket{101}-\ket{110}$ \\ 
\cline{2-2}\cline{4-4}
& $-\sigma_Z\otimes A_{\omega,\beta}\otimes A_{\omega,\beta-\pi} $  &  & $\ket{000}+\ket{011}$ \\ 

\cline{1-2}\cline{4-4}
\multirow{2}{*}{$-\sigma_Z\otimes A_{\varphi,0}\otimes A_{\pi-\varphi,0}$}& $-\sigma_Z\otimes A_{\omega,\beta}\otimes A_{\omega,\beta} $  &  & $\ket{101}+\ket{110}$\\ 
\cline{2-2}\cline{4-4}
&$-\sigma_Z\otimes A_{\omega,\beta}\otimes A_{\omega,-\beta} $  &  & $\ket{000}-\ket{011}$\\ 
\hline

\multirow{4}{*}{$\sigma_Z\otimes A_{\varphi,0}\otimes A_{\varphi,0}$} &$A_{\omega,0}\otimes \sigma_Z \otimes A_{\omega,0} $  & \multirow{6}{*}{$\varphi,\omega \neq 0,\pi$} & $\ket{000}+\ket{011}+\ket{101}-\ket{110}$ \\ 
\cline{2-2}\cline{4-4}
&$A_{\omega,0}\otimes \sigma_Z \otimes A_{\omega,\pi} $  &  & $\ket{000}+\ket{011}-\ket{101}+\ket{110}$ \\

\cline{2-2}\cline{4-4}
& $A_{\omega,0}\otimes -\sigma_Z \otimes A_{\pi-\omega,0} $  & & $\ket{000}+\ket{011}-\ket{101}+\ket{110}$ \\ 
\cline{2-2}\cline{4-4}
& $A_{\omega,0}\otimes -\sigma_Z \otimes A_{\pi-\omega,\pi} $  &  & $\ket{000}-\ket{011}-\ket{101}-\ket{110}$ \\ 

\cline{1-2}\cline{4-4}
\multirow{2}{*}{$-\sigma_Z\otimes A_{\varphi,0}\otimes A_{\pi-\varphi,0}$}& $A_{\omega,0}\otimes -\sigma_Z \otimes A_{\pi-\omega,0} $  & & $\ket{000}-\ket{011}-\ket{101}-\ket{110}$ \\ 
\cline{2-2}\cline{4-4}
& $A_{\omega,0}\otimes -\sigma_Z \otimes A_{\pi-\omega,\pi} $  &  & $\ket{000}-\ket{011}+\ket{101}+\ket{110}$ \\  
\hline
\end{tabular}
\endgroup
\caption{ Different ways three binary operators $\sigma_Z\otimes\sigma_Z\otimes \sigma_Z$, $A_{\theta_1,\varphi_1}\otimes A_{\theta_2,\varphi_2}\otimes A_{\theta_3,\varphi_3}$ and $A_{\theta_4,\varphi_4}\otimes A_{\theta_5,\varphi_5}\otimes A_{\theta_6,\varphi_6}$ can stabilize a unique state on three qubits. The first two columns give the expression of the second and third operators. This table shows all possibilities up to the permutation of the qubits.}
\label{tab}
\end{table}
    
The analysis of the case~\ref{case: 2} is conducted as follows. We apply the $2$-operators compatibility condition to the first two stabilizers, which allows us split the study into two cases. In the first case, the stabilizing operators are $\sigma_Z\otimes \sigma_Z \otimes \sigma_Z$, $A_{\theta}\otimes A_{\theta} \otimes \sigma_Z$ and $A_{\omega,\alpha}\otimes A_{\omega,\pm\alpha}\otimes \1$. For the $+\alpha$ choice, the stabilization can be treated with the already presented tools. For the choice $-\alpha$, we can verify that $(\ket{00}+\ket{11})\ket{0}$ is stabilized for all values of $\alpha$. This means that the stabilization is either non-unique, or exclusively fixes this particular state. Ultimately, we conclude that whatever the sign in front of $\alpha$, a uniquely stabilized state is equivalent to $(\ket{00}+\ket{11})\ket{0}$.\\

 For the second case, the operators are $\sigma_Z\otimes \sigma_Z \otimes \sigma_Z $, $A_{\theta}\otimes \sigma_Z \otimes A_{\theta}$ and $A_{\omega,\alpha}\otimes A_{\omega}\otimes \1$. We can apply local rotations to get a equivalent stabilization given by $A_\theta\otimes A_\theta \otimes A_\theta$, $A_{-\theta}\otimes A_\theta \otimes A_{-\theta}$ and $A_{\omega,\pi/2}\otimes A_{\omega,\pi/2}\otimes \1$. Applying the determinant method (see section~\ref{section: methods}) reveals that stabilization is only possible when certain conditions are met. Indeed the determinant $\det(M^\dagger M)=0$ can be factored into multiple terms. Asking that the determinant vanish imply one of the following condition 
 \begin{enumerate}
     \item $\theta=0,\pi/2,\pi$ \label{case: eq var 1}
     \item $\omega=\pi/2$ \label{case: eq var 2}
     \item $25 \cos (2 \theta -2 \omega )-4 \cos (4 \theta -2 \omega )+25 \cos (2 (\theta +\omega ))-4 \cos (4 \theta +2 \omega )+146 \cos (2 \theta )+8
   \cos (4 \theta )+54 \cos (2 \omega )-250=0$\label{case: eq var 3}
 \end{enumerate}
 
 We need to discard the cases $\theta=0,\pi/2,\pi$ (case~\ref{case: eq var 1}), since it would result in the first two operators being the same. By further analysis, using trigonometric expansion and substituting $\cos(\theta)^2\mapsto x$ and $\cos(\omega)^2\mapsto y$, it can be shown that the complex term (case~\ref{case: eq var 3}) is zero only when $\theta=0,\pi$ and $\omega=0,\pi$, which we should also exclude. Therefore, the only remaining case is $\omega=\pi/2$ (case~\ref{case: eq var 2}). To analyze this case, we reverse the local rotation and find that the state $\ket{000}+\ket{011}+\ket{101}+\ket{110}$ is uniquely stabilized if $\theta\neq 0,\pi/2,\pi$.

\end{document}